\documentclass[a4paper,UKenglish]{article}
\usepackage{authblk}

\usepackage[numbers,sort&compress]{natbib}

\usepackage{amstext,amssymb,amsmath,amsfonts,amsthm,dsfont}
\usepackage{xspace}
\usepackage{enumerate}
\usepackage{mathtools}
\usepackage{comment}
\usepackage{balance}
\usepackage{mathrsfs}
\usepackage{cleveref}
\usepackage[usenames,dvipsnames]{xcolor}
\usepackage{xargs}
\usepackage[colorinlistoftodos,prependcaption,textsize=tiny]{todonotes}
\usepackage{subcaption}
\usepackage{IEEEtrantools}

\newenvironment{cenv}{\begin{list}{}{%
			\setlength{\labelwidth}{1.5em}%
			\setlength{\leftmargin}{\labelwidth}%
			\addtolength{\leftmargin}{\labelsep}%
			\setlength{\listparindent}{0em}%
			\setlength{\topsep}{10pt}%
			\setlength{\itemsep}{5pt}%
			\setlength{\parsep}{0pt}%
		}
	}{
	\end{list}
}

\newcounter{claimcounter}

\newcounter{algorithmcounter}
\newenvironment{algorithm}{
	
	~\refstepcounter{algorithmcounter}
	\begin{cenv}
		\item[{\textbf{Algorithm \arabic{algorithmcounter}.}}]
	}{
	\end{cenv}
}

\newcommand{\N}{\mathbb{N}}

\newcommand{\Start}[1]{\mathscr{S}(#1)}
\newcommand{\End}[1]{\mathscr{E}(#1)}
\newcommand{\sizeof}[1]{\left|#1\right|}

\newcommand{\act}{\operatorname{active}}
\newcommand{\inact}{\operatorname{inactive}}

\newcommand{\updatesequence}{\mathfrak{R}}
\newcommand{\round}{\mathfrak{r}}

\newcommand{\blockset}{\mathcal{B}}

\newcommand{\roundr}{r}
\newcommand{\rounds}{\ell}

\newtheorem{theorem}{Theorem}

\newtheorem{lemma}[theorem]{Lemma}

\newtheorem{corollary}[theorem]{Corollary}

\theoremstyle{definition}
\newtheorem{definition}[theorem]{Definition}

\pgfdeclarelayer{bg}    
\pgfsetlayers{bg,main}  

\usetikzlibrary{snakes,calc,shapes}

\tikzset{
	position/.style args={#1:#2 from #3}{
		at=($(#3)+(#1:#2)$)
	}
}

\title{Short Schedules for Fast Flow Rerouting}

\begin{document}

\author[1]{Saeed Akhoondian Amiri\thanks{samiri@mpi-inf.mpg.de}}
\author[2]{Szymon Dudycz\thanks{szymon.dudycz@gmail.com}}
\author[3]{Mahmoud Parham\thanks{mahmoud.parham@univie.ac.at}}
\author[3]{Stefan Schmid\thanks{stefan\_schmid@univie.ac.at}}
\author[4]{Sebastian Wiederrecht\thanks {sebastian.wiederrecht@tu-berlin.de}}

\affil[1]{Max-Planck Institute for Informatics, Germany}
\affil[2]{University of Wroclaw, Poland}
\affil[3]{University of Vienna, Austria}
\affil[4]{TU Berlin, Germany}
\maketitle

\sloppy

\begin{abstract}
This paper studies the fundamental problem of how to reroute 
$k$ unsplittable flows of a certain demand in a capacitated network 
from their current paths to their respective new paths, in a congestion-free 
manner and \emph{fast}. 
This scheduling problem has applications in traffic engineering in communication networks 
and has recently received much attention in software-defined networks,
in which updates are distributed over an asynchronous network by a software controller. However, 
existing algorithms for this problem either have a super-polynomial runtime 
or only compute \emph{feasible} schedules, which do not provide any guarantees on 
the length of the rerouting schedule. 

This paper presents the first polynomial-time algorithm 
for computing shortest update schedules
to reroute flows
in a congestion-free manner. 
We contribute an almost tight characterization 
of the polynomial-time
tractability of the problem:
We present the first polynomial-time solution for this problem
for \emph{two} flows, but also show that even the question 
whether a \emph{feasible} update 
schedule exists, is already NP-hard for \emph{six} flows. 

In fact, the presented algorithm runs in \emph{linear time} 
and is hence not only optimal in terms of scheduling but also 
asymptotically optimal in terms of runtime. \end{abstract}


\section{Introduction}\label{sec:intro}

Emerging software-defined communication networks 
provide \emph{direct} and ``\emph{algorithmic}'' control over the 
forwarding rules of nodes (i.e., routers and switches) 
and hence the network \emph{routes}. 
The resulting routes are not restricted to follow
only shortest paths and moreover, 
they can be flexibly \emph{adapted over time}, e.g., depending
on certain events in the dataplane. 
Indeed, there are many reasons why flows may need to be 
\emph{rerouted}~\cite{update-survey}, including 
security and policy changes (e.g., suspicious traffic
is rerouted via a firewall), 
traffic engineering optimizations, reactions to changes in the demand, 
maintenance work, failures, etc. 

Implementing route changes however is challenging, 
since updating a route usually involves the distribution of
new (forwarding) rules across the asynchronous communication
network, and since even \emph{during} such route changes, 
it is important to maintain certain safety properties.
In particular, the routes of flows should be changed without
causing any congestion or introducing temporary forwarding loops.
For example, in a Software Defined Network (SDN), 
rules are communicated by the remote software controller. 
Therefore, updates have to be distributed in \emph{rounds}, 
in which switches acknowledge the next batch of updates~\cite{roger,sigmetrics16,ludwig2015scheduling}. 
 
This introduces a \emph{scheduling problem}: In which order
to update the different forwarding rules for the different flows and switches
over time, such that these safety properties are maintained at any time?
And how to schedule these updates such that the rerouting time
(and number of controller interactions) is minimized?

\subsection{A Simple Example}

Figure~\ref{fig:ex} gives an example
of the flow rerouting problem.  
We want to schedule 
the rerouting of $2$ flows in a $5$-node network,
connecting nodes $\{s,u,v,w,t\}$
with $7$ edges $\{\{s,u\},\{s,w\},\{u,w\},
\{u,v\},\{v,w\},\{v,t\},\{w,t\}\}$.
In this example, both flows originate at $s$
and end at $t$: The first flow is indicated
in \textcolor{red}{red} and the second
flow in \textcolor{blue}{blue}. 

\begin{figure}[t]
	\begin{center}
		\begin{tikzpicture}[scale=0.9]
		\tikzset{>=latex} 
		
		\node (o) [] {};
		
		\node (c-1-1) [position=0:0mm from o] {};

		\foreach\i in {1} {
			\foreach\j in {1} {
				
				\node (s-\i-\j) [fill,inner sep=0pt,minimum size=7pt,draw,circle,thick,scale=0.7,position=162:17mm from c-\i-\j] {};
				
				\node (s-\i-\j-1) [inner sep=0pt,position=330:0.72mm from s-\i-\j] {};
				\node (s-\i-\j-2) [inner sep=0pt,scale=0.75,position=270:0.72mm from s-\i-\j] {};
				\node (s-\i-\j-3) [inner sep=0pt,scale=0.75,position=90:0.72mm from s-\i-\j] {};
				\node (s-\i-\j-4) [inner sep=0pt,scale=0.75,position=30:0.72mm from s-\i-\j] {};
				
				\node (t-\i-\j) [fill,inner sep=0pt,minimum size=7pt,draw,circle,thick,scale=0.7,position=18:17mm from c-\i-\j] {};
				
				\node (t-\i-\j-1) [inner sep=0pt,scale=0.75,position=270:0.72mm from t-\i-\j] {};
				\node (t-\i-\j-2) [inner sep=0pt,scale=0.75,position=210:0.72mm from t-\i-\j] {};
				\node (t-\i-\j-3) [inner sep=0pt,scale=0.75,position=160:0.72mm from t-\i-\j] {};
				\node (t-\i-\j-4) [inner sep=0pt,scale=0.75,position=90:0.72mm from t-\i-\j] {};
				
				\node (u-\i-\j) [fill,inner sep=0pt,minimum size=7pt,draw,circle,thick,scale=0.5,position=224:10mm from c-\i-\j] {};
				
				\node (u-\i-\j-1) [inner sep=0pt,scale=0.75,position=90:0.72mm from u-\i-\j] {};
				\node (u-\i-\j-2) [inner sep=0pt,scale=0.75,position=180:0.72mm from u-\i-\j] {};
				\node (u-\i-\j-3) [inner sep=0pt,scale=0.75,position=120:0.72mm from u-\i-\j] {};
				\node (u-\i-\j-4) [inner sep=0pt,scale=0.75,position=0:0.72mm from u-\i-\j] {};	
				
				\node (v-\i-\j) [fill,inner sep=0pt,minimum size=7pt,draw,circle,thick,scale=0.5,position=316:10mm from c-\i-\j] {};
				
				\node (v-\i-\j-1) [inner sep=0pt,scale=0.75,position=90:0.72mm from v-\i-\j] {};
				\node (v-\i-\j-2) [inner sep=0pt,scale=0.75,position=180:0.72mm from v-\i-\j] {};
				\node (v-\i-\j-3) [inner sep=0pt,scale=0.75,position=60:0.72mm from v-\i-\j] {};
				\node (v-\i-\j-4) [inner sep=0pt,scale=0.75,position=0:0.72mm from v-\i-\j] {};		
				
				\node (w-\i-\j) [fill,inner sep=0pt,minimum size=7pt,draw,circle,thick,scale=0.5,position=90:15mm from c-\i-\j] {};
				
				\node (w-\i-\j-1) [inner sep=0pt,scale=0.75,position=320:0.72mm from w-\i-\j] {};
				\node (w-\i-\j-2) [inner sep=0pt,scale=0.75,position=240:0.72mm from w-\i-\j] {};
				\node (w-\i-\j-3) [inner sep=0pt,scale=0.75,position=180:0.72mm from w-\i-\j] {};
				\node (w-\i-\j-4) [inner sep=0pt,scale=0.75,position=30:0.72mm from w-\i-\j] {};	
				\node (w-\i-\j-5) [inner sep=0pt,scale=0.75,position=270:0.72mm from w-\i-\j] {};
				\node (w-\i-\j-6) [inner sep=0pt,scale=0.75,position=270:0.72mm from w-\i-\j] {};
				
				\pgfmathtruncatemacro{\ilabel}{\j+3*(\i-1)};
				
				\node (labels-\i-\j) [position=150:4mm from s-\i-\j] {$s$};	
				\node (labelt-\i-\j) [position=30:4mm from t-\i-\j] {$t$};
				\node (labelt-\i-\j) [position=270:4mm from u-\i-\j] {$u$};
				\node (labelt-\i-\j) [position=270:4mm from v-\i-\j] {$v$};
				\node (labelt-\i-\j) [position=90:4mm from w-\i-\j] {$w$};
				
			}	
		}
		
		\begin{pgfonlayer}{bg}    
		
		\draw [cyan,line width=6.5pt,line cap=round,opacity=0.2] (s-1-1) to (w-1-1);
		\draw [cyan,line width=6.5pt,line cap=round,opacity=0.2] (w-1-1) to (t-1-1);
		
		\draw [Magenta,line width=6.5pt,line cap=round,opacity=0.2] (s-1-1) to (u-1-1);
		\draw [Magenta,line width=6.5pt,line cap=round,opacity=0.2] (u-1-1) to (v-1-1);
		\draw [Magenta,line width=6.5pt,line cap=round,opacity=0.2] (v-1-1) to (t-1-1);
		
		\draw [blue,line width=1.3pt,->] (s-1-1-4) to (w-1-1-2);
		\draw [blue,line width=1.3pt,->] (w-1-1-1) to (t-1-1-3);
		
		\draw [red,line width=1.3pt,->] (s-1-1-1) to (u-1-1-3);
		\draw [red,line width=1.3pt,->] (u-1-1-4) to (v-1-1-2);
		\draw [red,line width=1.3pt,->] (v-1-1-3) to (t-1-1-2);
		
		\draw [blue,line width=1.3pt,->,dashed] (s-1-1-2) to (u-1-1-2);
		\draw [blue,line width=1.3pt,->,dashed] (u-1-1-1) to (w-1-1-6);
		\draw [blue,line width=1.3pt,->,dashed] (w-1-1-5) to (v-1-1-1);
		\draw [blue,line width=1.3pt,->,dashed] (v-1-1-4) to (t-1-1-1);
		
		\draw [red,line width=1.3pt,->,dashed] (s-1-1-3) to (w-1-1-3);
		\draw [red,line width=1.3pt,->,dashed] (w-1-1-4) to (t-1-1-4);
		
		\node (l-sw-1-1) [position=130:17.5mm from c-1-1] {$\frac{1}{1}$};
		\node (l-wt-1-1) [position=50:17.5mm from c-1-1] {$\frac{1}{2}$};
		\node (l-uw-1-1) [position=160:7mm from c-1-1] {$\frac{0}{1}$};
		\node (l-wv-1-1) [position=20:7mm from c-1-1] {$\frac{0}{1}$};
		\node (l-uv-1-1) [position=270:12mm from c-1-1] {$\frac{1}{1}$};
		\node (l-su-1-1) [position=200:15mm from c-1-1] {$\frac{1}{2}$};
		\node (l-vt-1-1) [position=340:15mm from c-1-1] {$\frac{1}{1}$};

		\end{pgfonlayer}

		\end{tikzpicture}	
	\end{center}
	\vspace{-1em}
	\caption{The flow rerouting problem: Example.} 
	\label{fig:ex}
\end{figure}
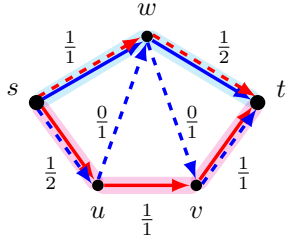

Each of the two flows has an \emph{original (``old'') route} and a \emph{new route},
which it should be updated to. 
We indicate the original route 
with a \emph{solid} 
line and the new route with a \emph{dotted} line.
For example, the  original route of the red flow is 
$(s,u,v,t)$ and needs to be updated to $(s,w,t)$.
The  original route of the blue flow is 
$(s,w,t)$ and needs to be updated to $(s,u,w,v,t)$.

In other words, each flow defines an 
\emph{update pair}, consisting of two routes (the original and the new one):
Acordingly, updates are denoted using tuples, i.e.,
$(v,\textcolor{blue}{B})$ means that 
we activate all 
inactive (dotted) 
outgoing blue edges (the new \emph{forwarding rules}) of vertex $v$ and deactivate all of its active 
	(solid) outgoing edges (the old \emph{forwarding rules}).

In this example, we assume that both flows consume 
$1$ unit of bandwidth on each link they traverse.
Both flows are unsplittable.
Accordingly, we annotate the network edges in the figure 
with two numbers,
$\frac{x}{y}$, where $x$ denotes the bandwidth
consumed  by the two flows on the corresponding edge 
\emph{before} rerouting and $y$ denotes the edge capacity.

How to reroute the two flows from their old paths to their
new paths in a congestion-free manner?
In this example, initially, 
we cannot perform the update $(s,\textcolor{red}{R})$, 
since the edge $(s,w)$ has a capacity of $1$ which is currently
used by the blue flow. So 
	the first part of an update schedule could look like this,
	where the updates in this sequence are performed one-by-one:
		$$
		 (u,\textcolor{blue}{B}) ,  (s,\textcolor{blue}{B}) ,  (s,\textcolor{red}{R}) , \dots
		$$
	In this case the red flow would be 
	routed along the edge $(s,w)$.
	However, from there, it could not reach $t$ anymore,
	after performing the update $(s,\textcolor{red}{R})$:
	the schedule is invalid. In fact,
	no valid update sequence can start like 
	the example above.
	
	One valid sequence is the following:
		$$
		 (u,\textcolor{blue}{B}) ,  (s,\textcolor{blue}{B}) ,  (w,\textcolor{red}{R}) ,  
		(s,\textcolor{red}{R}) ,  (u,\textcolor{red}{R}) ,  (v,\textcolor{red}{R}) ,  
		(v,\textcolor{blue}{B}) ,  (w,\textcolor{blue}{B})
		$$
	But this schedule requires $8$ rounds, 
	updating only one vertex for one flow at a time. 
	
	A faster update sequence 
	schedules multiple updates in a single round, if 
	possible without introducing congestion: Updates that are 
	scheduled for the same round are asynchronous and 
	can occur in any order,
	and hence, need to be performed carefully. 
The following schedule requires 4 rounds and 
is the shortest valid congestion-free
flow rerouting solution for our example:	
		$$
		(u,\textcolor{blue}{B}) ,  \{(s,\textcolor{blue}{B}) ,  
		(v,\textcolor{blue}{B}) ,  (w,\textcolor{red}{R})\} ,  
		(s,\textcolor{red}{R}) ,  \{(u,\textcolor{red}{R}) ,  (v,\textcolor{red}{R}) ,  (w,\textcolor{blue}{B})\}
		$$

A rigorous formal model for this
problem will be given later in this paper.

\subsection{Our Contributions}

This paper initiates the study of polynomial-time
scheduling algoritms to reroute flows in a congestion-free manner
and \emph{fast}. In particular, we contribute the, 
to the best of our knowledge \emph{first},
polynomial-time algorithm to compute 
shortest rerouting schedules \emph{for two flows}.
In fact, our algorithm runs in (deterministic) \emph{linear time};
 its runtime is hence asymptotically optimal.
Moreover, our algorithm is elegant.

We show that this is almost as good as one can hope for
when investing only polynomial time algorithms:
we rigorously prove that even \emph{deciding}
whether a congestion-free reroute schedule exists
is NP-hard, already for 6 flows.
In other words, we provide an almost tight characterization
of the polynomial-time solvability of the problem.

In addition to our formal results, we also show empirically
that the schedules produced by our algorithm are significantly
shorter than the state-of-the-art algorithms focusing on \emph{feasibility}~\cite{icalp18}.




\subsection{Paper Organization}

The remainder of this paper is organized
as follows. Section~\ref{sec:model} presents
a formal model for the problem studied in this paper.
Section~\ref{2flows} describes and analyzes
a polynomial-time update scheduling algorithm 
for two flows and Section~\ref{sec:np-hard} 
presents the hardness proof for six flows. 
We present a nonpolynomial-time algorithm to compute
optimal schedules in Section~\ref{sec:mip}
and present simulation results in Section~\ref{sec:sims}.
After reviewing related work in Section~\ref{sec:relwork},
we conclude our contribution in Section~\ref{sec:conclusion}.

\section{A Rigorous Formal Model}\label{sec:model}

This section presents a rigorous formal model for the
fast congestion-free flow rerouting problem 
introduced intuitively in Figure~\ref{fig:ex}.
The problem can be described in terms of edge capacitated directed graphs. 
In what follows, we will assume basic familiarity with directed graphs and we refer the
reader to~\cite{digraphs} for more details. We denote a directed
edge $e$ with head $v$ and tail $u$ by $e=(u,v)$. For an undirected
edge $e$ between vertices $u,v$, we write $e=\{u,v\}$; 
$u,v$ are called
endpoints of $e$.

For ease of presentation and without loss of generality, we consider directed graphs
with only one source vertex (where flows will originate) and one terminal
vertex (the flows' sink). We call this graph a \emph{flow network}.
The forwarding rules that define the paths considered in 
our problem, are best seen as flows in a network. 
We will be interested in rerouting flows such that natural 
notions of consistency are preserved, such as loop-freedom
and congestion-freedom. 
In particular, we will say that a set of
  flows is valid if the edge capacities 
of the underlying network are respected.
\begin{definition}[Flow Network, Flow, Valid Flow Sets]
A \textbf{flow network} is a directed capacitated graph $G=(V,E,s,t,c)$, 
where $s$ is the \emph{source}, $t$ the \emph{terminal}, $V$ is the set 
of vertices with $s,t\in V$, $E\subseteq V\times V$ is a set of ordered 
pairs known as edges, and $c\colon E\rightarrow\N$ a 
capacity function assigning a capacity $c(e)$ to every edge $e\in E$.
An \emph{$(s,t)$-flow} $F$ of capacity $d\in\N$ is a 
\emph{directed path} from $s$ to $t$ in a flow network such that 
$d\leq c(e)$ for all $e\in E(F)$. Given a  
$\mathcal{F}$ of $(s,t)$-flows $F_1,\dots,F_k$ with demands 
$d_1,\dots,d_k$ respectively, we call $\mathcal{F}$ a \textbf{valid flow
set}, or simply \textbf{valid}, if $c(e)\geq\sum_{i\colon e\in E(F_i)}d_i$.
\end{definition}
Recall that we consider the problem of how to reroute a current (old) flow to a 
new flow, and hence we will consider such flows in ``update pairs'': 
\begin{definition}[Update Flow Pair] 
	An \textbf{update flow pair} $P=(F^o,F^u)$ consists 
	of two $(s,t)$-flows $F^o$, the \emph{old flow}, and $F^u$, 
	the \emph{update (or new) flow}, each of demand $d$.
\end{definition}
The update flow network is a flow network  (the underlying edge capacitated graph)  
together with a \textbf{valid} family of flow pairs. For an illustration, recall 
the initial network in Figure~\ref{fig:ex}: 
The old flows are presented as the directed paths made of solid edges 
and the new ones are represented by the dashed edges. 

A flow can be rerouted by updating the outgoing edges of the vertices along its path
(the forwarding rules), i.e., 
by blocking the outgoing edge of the old flow and by allowing
traffic along the outgoing edge of the new flow (if either of them
exists). If these two edges coincide, there are no changes. 
In order to ensure transient consistency, the updates of these 
outgoing edges need to be scheduled 
over time: this results in a sequence which can be partitioned into
update rounds.

\begin{definition}[Resolving Updates, Update Sequence]
Given $G=(V,E,\mathcal{P},s,t,c)$ and an update flow pair $P=(F^o,F^u)\in\mathcal{P}$ 
of demand $d$, we consider the \textbf{activation label} $\alpha_P\colon E(F^o\cup F^u)\times 2^{V\times\mathcal{P}}\rightarrow \left\{ \act,\inact \right\}$.
For an edge $(u,v)\in E(F^o\cup F^u)$ and a set of updates $U\subseteq V\times\mathcal{P}$, $\alpha_P$ 
is defined as follows:
\[
\alpha_P((u,v),U)=\left\{\begin{array}{ll} \act, &
                                                   \text{if}~(u,P)\notin
                                                   U, (u,v)\in
                                                   E(F^o), \\\act, &
                                                                     \text{if}~(u,P)\in
                                                                     U, (u,v)\in
                                                                     E(F^u),\\
                           \inact, & \text{otherwise.}\end{array}\right.
\]	
The graph: 
$$\alpha(U,G)=(V,\left\{ e\in E \left|\exists i\in[k]~\text{s.t.}~\alpha_{P_i}(e,U)=\act
  \right.\right\})$$ is called the \textbf{$U$-state} of $G$ and we
call any update in $U$ \textbf{resolved}.

An \textbf{update sequence} $\updatesequence=(\round_1,\dots,\round_{\ell})$ is an ordered
partition of $V\times\mathcal{P}$. For every such $i$ we define $U_i=\bigcup_{j=1}^i\round_i$
and consider the activation label $\alpha_P^i(e)=\alpha_P(e,U_i)$ for
every update flow pair $P=(F^o,F^u)\in\mathcal{P}$ of demand $d$ and edge $e\in E(F^o\cup F^u)$.
\end{definition}


Let $(u,P)$ be some update. When we say that we want to \emph{resolve}
$(u,P)$, we mean that we target a state of $G$ in which
$(u,P)$ is resolved. In most cases this will mean to add $(u,P)$ to
the set of already resolved updates. 
With a slight abuse of notation, let define $\alpha_P(U,G) = (V(F^o)\cup V(F^u), (E(F^o) \cup E(F^u))$.

In the definition of an update sequence, $\round_i$ for $i\in[\ell]$ is a \textbf{\emph{round}}. 
We define the \emph{initial round} $\round_0=\emptyset$.
Recall that we consider unsplittable flows which travel along a single path. 
The following will clarify 
how active edges are to be used. 
\begin{definition}[Transient Flow, Transient Family]
The flow pair $P$ is called \textbf{transient} for some set of updates $U\subseteq V\times\mathcal{P}$, 
if $\alpha_P(U,G)$ contains a unique valid $(s,t)$-flow $T_{P,U}$. 
If there is a \textbf{valid} family $\mathcal{P}=\left\{ P_1,\dots P_k \right\}$ of update flow pairs with 
demands $d_1,\dots,d_k$ respectively, we call $\mathcal{P}$ a 
\textbf{transient family} for a set of updates $U\subseteq
V\times\mathcal{P}$, if and only if every $P\in\mathcal{P}$ is transient for $U$.
\end{definition}

In short, the transient flows look
like a path of active edges for flow $F$, which starts at the source
vertex and ends at the terminal vertex. Note that there may be some
active edges connected to this path, but they cannot be used to route
the flow since $T_{P,U}$ is unique after resolving
$U$.
The collection of the transient flows corresponding to the transient family 
is a snapshot of a valid
updating scenario.
Whenever we say a path $p$ ``routes'' a flow $F$,
we mean that
all edges of path $p$ are active for flow $F$.
%

In each round $\round_i$, any subset of updates of $\round_i$ resolved without
considering the remaining updates of $\round_i$ should allow a
transient flow for every flow pair. This models the asynchronous
nature of the implementation of the update commands in each round.
\begin{definition}[Consistency Rule] 
Let $\updatesequence=(\round_1,\ldots,\round_\ell)$ be an update sequence and $i\in[\rounds]$.
We require that for any $S\subseteq\round_i,\mathcal{U}_i^S\coloneqq S\cup \bigcup_{i-1} \round_j$,
there is a family of transient flow pairs.
\end{definition}
\begin{definition}[Valid Update]
An update sequence $\updatesequence$ is \textbf{valid}, or \textbf{feasible}, 
if every round $\round_i\in\updatesequence$ obeys the consistency rule.
\end{definition}
%

Note that we do not forbid any edge $e\in E(F^o_i\cap F^u_i)$ and 
we never activate or deactivate such an edge. Starting with an initial update 
flow network, these edges will be active and remain so until all updates are resolved. 
Hence there are vertices $v\in V$ with either no outgoing edge for a given flow pair $F$ at all;
or $v$ has an outgoing edge, but this edge is used by both the old and the update flow of $F$. 
We will call such updates $(v,P)$ \emph{empty}.

Empty updates do not have any impact on the actual problem since they
never change any transient flow. Hence they can always be scheduled in
the first round and thus w.l.o.g.~we can ignore them in the
following. Let us now define the main problem which we
consider in this paper.
\begin{definition}[\textsc{$k$-Network Flow Update Problem}]
Given an update flow network $G$ with $k$ update flow pairs, 
is there a feasible update sequence $\updatesequence$? The corresponding optimization 
problem is: What is the minimum $\ell$ such that there exists a valid update sequence 
$\updatesequence$ using exactly $\ell$ rounds?
\end{definition}

Finally, we introduce some \textbf{preliminaries}.
Let $G=(V,E,\mathcal{P},s,t,c)$ be an update flow network consisting
of two flow pairs $P^1,P^2$, such that 
each flow pair is an acyclic graph. For a flow pair $P^i$ ($i\in \{1,2\}$),
let $\prec^i$ be a topological 
order on its vertices $V=\left\{ v_1,\dots,v_n \right\}$. We may
write $\prec$ for $\prec^i$ whenever $i$ is clear from the context. 

The following applies to both feasible and shortest schedules,
and hence, we use the same terminology as Amiri et
al.~\cite{icalp18}. We only slightly modify the terminology 
as unlike Amiri et al., we do not require that the sum of updates forms a DAG
(but only the pairs). 
Let $P_i=(F^o_i,F^u_i)$ be an update flow 
pair of demand $d$. 
We define a topological order $v_1^i,\dots,v_{\ell_i^o}^i$
on the vertices of $F^o_i$ w.r.t.~$\prec^i$ (recall
that we $P_i$ forms an acyclic graph); analogously,
let $u_1^i,\dots,v_{\ell^u_i}^i$ 
be the order on $F^u_i$. Furthermore, let 
$V(F^o_i)\cap V(F^u_i)=\left\{ z_1^i,\dots,z^i_{k_i} \right\}$ 
be ordered by $\prec^i$ as well. 
The subgraph of $F_i^o\cup F_i^u$ induced by the set $\left\{ v\in
  V(F_i^o\cup F_i^u) ~|~ z_j^i \prec v \prec z^i_{j+1} \right\}$,
$j\in[k_i-1]$, is called the $j$th \emph{block} of the 
update flow
pair $F_i$, or simply the $j$th \emph{$i$-block}. We 
will denote this block by $b^i_j$.

For a block $b$, we define $\Start{b}$ to be 
the \emph{start of the block}, i.e., the smallest vertex
w.r.t.~$\prec^i$; similarly, $\End{b}$ is the \emph{end of the block}:
the largest vertex w.r.t.~$\prec^i$. 

Let $G=(V,E,\mathcal{P},s,t,c)$ be an update flow network with
$\mathcal{P}=\left\{ P_1,\dots,P_k \right\}$ and let
$\blockset$ be the set of its
blocks. We define a binary relation $<$ between two blocks as follows. 
For two blocks $b_1,b_2\in \blockset$, where $b_1$ is an $i$-block and $b_2$ a
$j$-block, $i,j\in[k]$, we say $b_1<b_2$ ($b_1$ \emph{is smaller than}
$b_2$) if one of the following holds.
\begin{enumerate}[i]
\item $\Start{b_1} \prec \Start{b_2}$,
\item if $\Start{b_1}=\Start{b_2}$ then $b_1<b_2$, if $\End{b_1} \prec \End{b_2}$,
\item if $\Start{b_1}=\Start{b_2}$ and $\End{b_1}=\End{b_2}$ then $b_1<b_2$, if $i<j$.
\end{enumerate}
Let $b$ be an $i$-block and $P_i$ the corresponding update flow pair. 
For a feasible update sequence $\updatesequence$, we will denote the round 
$\updatesequence(\Start{b},P_i)$ by $\updatesequence(b)$. We say that $i$-block $b$ is
\emph{updated}, if all edges in $b\cap F^u_i$ are active and all edges
in $b\cap F_i^o\setminus F_i^u$ are inactive.

\section{A Fast Scheduling Algorithm}\label{2flows}

This section presents an elegant, \emph{linear-time} and deterministic 
and deterministic algorithm to compute
shortest update schedules for two flows. 

Let $G=(V,E,\mathcal{P},s,t,c)$ be an update 
flow network where $(V,E)$ is the union of the
DAGs implied by the flow pairs.
Let $\mathcal{P}=\left\{ B,R \right\}$ be the two update flow pairs 
with $B=(B^o,B^u)$ and $R=(R^o,R^u)$ of demands $d_B$ and $d_R$. 
As in the previous section, we identify $B$ with blue and $R$ with red.

We say that an $I$-block $b_1$ is \emph{dependent} 
on a $J$-block $b_2$, $I,J\in\left\{ B,R \right\}$, $I\neq J$, if 
there is an edge $e\in (E(b_1)\cap E(I^u))\cap(E(b_2)\cap E(J^o))$, 
but $c(e) < d_I+d_J$. In fact, to update $b_1$, we either 
violate capacity constraints, or we update $b_2$ first in order to 
prevent congestion. In this case, we write $b_1\rightarrow b_2$ 
and say that $b_1$ \emph{requires} $b_2$. A block that does not depend on any other block is called \emph{free}.

We say a block $b$ is a \emph{free block}, if it is not dependent on any 
other block. A \emph{dependency graph} of $G$ is a graph 
$D=(V_D,E_D)$ for which there exists a bijective 
mapping $\mu\colon V(D)\leftrightarrow B(G)$, and there is an 
edge $(v_b,v_{b'})$ in $D$ if $b\rightarrow b'$. Clearly, a block $b$ 
is free if and only if it corresponds to a sink in $D$.

\smallskip

We propose the following algorithm to check the feasibility 
of the flow rerouting problem. 

\begin{algorithm}\textbf{Feasible $2$-Flow DAG Update}\label[algorithm]{alg:main} 
	\begin{enumerate}
		\item [] \textbf{Input: Update Flow Network $G$}
		\item Compute the dependency graph $D$ of $G$.
		\item If there is a cycle in $D$, return \emph{impossible to update}.
		\item While $D\neq \emptyset$ repeat:
		\begin{enumerate}[i]
			\item \label{lbl:step} Update all blocks which correspond to the sink 
			vertices of $D$.
			\item Delete all of the current sink vertices from $D$.
		\end{enumerate}
	\end{enumerate}
\end{algorithm}

Recall that empty updates can always be scheduled in the first round,
even for infeasible problem instances. So
for \Cref{alg:main} and all following algorithms, we simply
assume these updates to be scheduled together with the non-empty updates of round $1$.

Figure~\ref{fig:exwithblocks} gives an example of an update flow
network on a DAG and illustrates the block decomposition and 
its value to finding a feasible update sequence.

\begin{figure*}[t]
	\begin{center}
		\begin{tikzpicture}[scale=0.8]
		\tikzset{>=latex} 
		
		\node (o) [] {};
		\node (u) [position=270:55mm from o] {};

		\node (c-1-1) [position=180:55mm from o] {};
		\node (c-1-2) [position=0:0mm from o] {};
		\node (c-1-3) [position=0:55mm from o] {};
		\node (c-2-1) [position=180:55mm from u] {};
		\node (c-2-2) [position=0:0mm from u] {};
		\node (c-2-3) [position=0:55mm from u] {};

		\foreach\i in {1,...,2} {
			\foreach\j in {1,...,3} {
				
				\node (s-\i-\j) [fill,inner sep=0pt,minimum size=7pt,draw,circle,thick,scale=0.7,position=162:17mm from c-\i-\j] {};
				
				\node (s-\i-\j-1) [inner sep=0pt,position=330:0.72mm from s-\i-\j] {};
				\node (s-\i-\j-2) [inner sep=0pt,scale=0.75,position=270:0.72mm from s-\i-\j] {};
				\node (s-\i-\j-3) [inner sep=0pt,scale=0.75,position=90:0.72mm from s-\i-\j] {};
				\node (s-\i-\j-4) [inner sep=0pt,scale=0.75,position=30:0.72mm from s-\i-\j] {};
				
				\node (t-\i-\j) [fill,inner sep=0pt,minimum size=7pt,draw,circle,thick,scale=0.7,position=18:17mm from c-\i-\j] {};
				
				\node (t-\i-\j-1) [inner sep=0pt,scale=0.75,position=270:0.72mm from t-\i-\j] {};
				\node (t-\i-\j-2) [inner sep=0pt,scale=0.75,position=210:0.72mm from t-\i-\j] {};
				\node (t-\i-\j-3) [inner sep=0pt,scale=0.75,position=160:0.72mm from t-\i-\j] {};
				\node (t-\i-\j-4) [inner sep=0pt,scale=0.75,position=90:0.72mm from t-\i-\j] {};
				
				\node (u-\i-\j) [fill,inner sep=0pt,minimum size=7pt,draw,circle,thick,scale=0.5,position=224:10mm from c-\i-\j] {};
				
				\node (u-\i-\j-1) [inner sep=0pt,scale=0.75,position=90:0.72mm from u-\i-\j] {};
				\node (u-\i-\j-2) [inner sep=0pt,scale=0.75,position=180:0.72mm from u-\i-\j] {};
				\node (u-\i-\j-3) [inner sep=0pt,scale=0.75,position=120:0.72mm from u-\i-\j] {};
				\node (u-\i-\j-4) [inner sep=0pt,scale=0.75,position=0:0.72mm from u-\i-\j] {};	
				
				\node (v-\i-\j) [fill,inner sep=0pt,minimum size=7pt,draw,circle,thick,scale=0.5,position=316:10mm from c-\i-\j] {};
				
				\node (v-\i-\j-1) [inner sep=0pt,scale=0.75,position=90:0.72mm from v-\i-\j] {};
				\node (v-\i-\j-2) [inner sep=0pt,scale=0.75,position=180:0.72mm from v-\i-\j] {};
				\node (v-\i-\j-3) [inner sep=0pt,scale=0.75,position=60:0.72mm from v-\i-\j] {};
				\node (v-\i-\j-4) [inner sep=0pt,scale=0.75,position=0:0.72mm from v-\i-\j] {};		
				
				\node (w-\i-\j) [fill,inner sep=0pt,minimum size=7pt,draw,circle,thick,scale=0.5,position=90:15mm from c-\i-\j] {};
				
				\node (w-\i-\j-1) [inner sep=0pt,scale=0.75,position=320:0.72mm from w-\i-\j] {};
				\node (w-\i-\j-2) [inner sep=0pt,scale=0.75,position=240:0.72mm from w-\i-\j] {};
				\node (w-\i-\j-3) [inner sep=0pt,scale=0.75,position=180:0.72mm from w-\i-\j] {};
				\node (w-\i-\j-4) [inner sep=0pt,scale=0.75,position=30:0.72mm from w-\i-\j] {};	
				\node (w-\i-\j-5) [inner sep=0pt,scale=0.75,position=270:0.72mm from w-\i-\j] {};
				\node (w-\i-\j-6) [inner sep=0pt,scale=0.75,position=270:0.72mm from w-\i-\j] {};
				
				\pgfmathtruncatemacro{\ilabel}{\j+3*(\i-1)};
				
				\node (labels-\i-\j) [position=150:4mm from s-\i-\j] {$s$};	
				\node (labelt-\i-\j) [position=30:4mm from t-\i-\j] {$t$};
				\node (labelt-\i-\j) [position=270:4mm from u-\i-\j] {$u$};
				\node (labelt-\i-\j) [position=270:4mm from v-\i-\j] {$v$};
				\node (labelt-\i-\j) [position=90:4mm from w-\i-\j] {$w$};
				
			}	
		}

		\node (labelf-1-1) [position=270:25mm from c-1-1] {initial network $G$};
		\node (dependency) [position=270:19mm from c-1-2] {\textcolor{blue}{$b_2$}$\rightarrow$\textcolor{red}{$r_1$}$\rightarrow$\textcolor{blue}{$b_1$}};	
		\node (labelf-1-2) [position=270:25mm from c-1-2] {block dependency};
		\node (labelf-1-3) [position=270:25mm from c-1-3] {prepare blocks};
		\node (labelf-2-1) [position=270:19mm from c-2-1] {update $b_1$};
		\node (labelf-2-2) [position=270:19mm from c-2-2] {update $r_1$};
		\node (labelf-2-3) [position=270:19mm from c-2-3] {update $b_2$};
		
		\begin{pgfonlayer}{bg}    
		
		\draw [cyan,line width=6.5pt,line cap=round,opacity=0.2] (s-1-1) to (w-1-1);
		\draw [cyan,line width=6.5pt,line cap=round,opacity=0.2] (w-1-1) to (t-1-1);
		
		\draw [Magenta,line width=6.5pt,line cap=round,opacity=0.2] (s-1-1) to (u-1-1);
		\draw [Magenta,line width=6.5pt,line cap=round,opacity=0.2] (u-1-1) to (v-1-1);
		\draw [Magenta,line width=6.5pt,line cap=round,opacity=0.2] (v-1-1) to (t-1-1);
		
		\draw [blue,line width=1.3pt,->] (s-1-1-4) to (w-1-1-2);
		\draw [blue,line width=1.3pt,->] (w-1-1-1) to (t-1-1-3);
		
		\draw [red,line width=1.3pt,->] (s-1-1-1) to (u-1-1-3);
		\draw [red,line width=1.3pt,->] (u-1-1-4) to (v-1-1-2);
		\draw [red,line width=1.3pt,->] (v-1-1-3) to (t-1-1-2);
		
		\draw [blue,line width=1.3pt,->,dashed] (s-1-1-2) to (u-1-1-2);
		\draw [blue,line width=1.3pt,->,dashed] (u-1-1-1) to (w-1-1-6);
		\draw [blue,line width=1.3pt,->,dashed] (w-1-1-5) to (v-1-1-1);
		\draw [blue,line width=1.3pt,->,dashed] (v-1-1-4) to (t-1-1-1);
		
		\draw [red,line width=1.3pt,->,dashed] (s-1-1-3) to (w-1-1-3);
		\draw [red,line width=1.3pt,->,dashed] (w-1-1-4) to (t-1-1-4);
		
		\node (l-sw-1-1) [position=130:17.5mm from c-1-1] {$\frac{1}{1}$};
		\node (l-wt-1-1) [position=50:17.5mm from c-1-1] {$\frac{1}{2}$};
		\node (l-uw-1-1) [position=160:7mm from c-1-1] {$\frac{0}{1}$};
		\node (l-wv-1-1) [position=20:7mm from c-1-1] {$\frac{0}{1}$};
		\node (l-uv-1-1) [position=270:12mm from c-1-1] {$\frac{1}{1}$};
		\node (l-su-1-1) [position=200:15mm from c-1-1] {$\frac{1}{2}$};
		\node (l-vt-1-1) [position=340:15mm from c-1-1] {$\frac{1}{1}$};

		\draw [Dandelion,line width=6.5pt,line cap=round,opacity=0.4] (s-1-2) to (w-1-2);
		\draw [Dandelion,line width=6.5pt,line cap=round,opacity=0.4] (v-1-2) to (t-1-2);
		
		\draw [blue,line width=1pt,->,decoration = {zigzag,segment length = 2mm, amplitude = 0.3mm}, decorate] (s-1-2-4) to (w-1-2-2);
		\draw [blue,line width=1pt,->] (w-1-2-1) to (t-1-2-3);
		
		\draw [red,line width=1pt,->] (s-1-2-1) to (u-1-2-3);
		\draw [red,line width=1pt,->] (u-1-2-4) to (v-1-2-2);
		\draw [red,line width=1pt,->] (v-1-2-3) to (t-1-2-2);
		
		\draw [blue,line width=1pt,->,dashed,decoration = {zigzag,segment length = 2mm, amplitude = 0.3mm}, decorate] (s-1-2-2) to (u-1-2-2);
		\draw [blue,line width=1pt,->,dashed,decoration = {zigzag,segment length = 2mm, amplitude = 0.3mm}, decorate] (u-1-2-1) to (w-1-2-6);
		\draw [blue,line width=1pt,->,dashed] (w-1-2-5) to (v-1-2-1);
		\draw [blue,line width=1pt,->,dashed] (v-1-2-4) to (t-1-2-1);
		
		\draw [red,line width=1pt,->,dashed] (s-1-2-3) to (w-1-2-3);
		\draw [red,line width=1pt,->,dashed] (w-1-2-4) to (t-1-2-4);
		
		\node (l-sw-1-2) [position=130:17.5mm from c-1-2] {$\frac{1}{1}$};
		\node (l-uw-1-2) [position=160:8mm from c-1-2] {$b_1$};
		\node (l-wv-1-2) [position=20:8mm from c-1-2] {$b_2$};
		\node (l-uv-1-2) [position=270:12mm from c-1-2] {$r_1$};
		\node (l-vt-1-2) [position=340:15mm from c-1-2] {$\frac{1}{1}$};

		\draw [cyan,line width=6.5pt,line cap=round,opacity=0.2] (s-1-3) to (w-1-3);
		\draw [cyan,line width=6.5pt,line cap=round,opacity=0.2] (w-1-3) to (t-1-3);
		
		\draw [Magenta,line width=6.5pt,line cap=round,opacity=0.2] (s-1-3) to (u-1-3);
		\draw [Magenta,line width=6.5pt,line cap=round,opacity=0.2] (u-1-3) to (v-1-3);
		\draw [Magenta,line width=6.5pt,line cap=round,opacity=0.2] (v-1-3) to (t-1-3);
		
		\draw [blue,line width=1.3pt,->] (s-1-3-4) to (w-1-3-2);
		\draw [blue,line width=1.3pt,->] (w-1-3-1) to (t-1-3-3);
		
		\draw [red,line width=1.3pt,->] (s-1-3-1) to (u-1-3-3);
		\draw [red,line width=1.3pt,->] (u-1-3-4) to (v-1-3-2);
		\draw [red,line width=1.3pt,->] (v-1-3-3) to (t-1-3-2);
		
		\draw [blue,line width=1.3pt,->,dashed] (s-1-3-2) to (u-1-3-2);
		\draw [blue,line width=1.3pt,->] (u-1-3-1) to (w-1-3-6);
		\draw [blue,line width=1.3pt,->,dashed] (w-1-3-5) to (v-1-3-1);
		\draw [blue,line width=1.3pt,->] (v-1-3-4) to (t-1-3-1);
		
		\draw [red,line width=1.3pt,->,dashed] (s-1-3-3) to (w-1-3-3);
		\draw [red,line width=1.3pt,->] (w-1-3-4) to (t-1-3-4);
		
		\node (l-sw-1-3) [position=130:17.5mm from c-1-3] {$\frac{1}{1}$};
		\node (l-wt-1-3) [position=50:17.5mm from c-1-3] {$\frac{1}{2}$};
		\node (l-uw-1-3) [position=160:7mm from c-1-3] {$\frac{0}{1}$};
		\node (l-wv-1-3) [position=20:7mm from c-1-3] {$\frac{0}{1}$};
		\node (l-uv-1-3) [position=270:12mm from c-1-3] {$\frac{1}{1}$};
		\node (l-su-1-3) [position=200:15mm from c-1-3] {$\frac{1}{2}$};
		\node (l-vt-1-3) [position=340:15mm from c-1-3] {$\frac{1}{1}$};

		\draw [cyan,line width=6.5pt,line cap=round,opacity=0.2] (s-2-1) to (u-2-1);
		\draw [cyan,line width=6.5pt,line cap=round,opacity=0.2] (u-2-1) to (w-2-1);
		\draw [cyan,line width=6.5pt,line cap=round,opacity=0.2] (w-2-1) to (t-2-1);
		
		\draw [Magenta,line width=6.5pt,line cap=round,opacity=0.2] (s-2-1) to (u-2-1);
		\draw [Magenta,line width=6.5pt,line cap=round,opacity=0.2] (u-2-1) to (v-2-1);
		\draw [Magenta,line width=6.5pt,line cap=round,opacity=0.2] (v-2-1) to (t-2-1);
		
		\draw [blue,line width=0.9pt,->,dashed] (s-2-1-4) to (w-2-1-2);
		\draw [blue,line width=1.3pt,->] (w-2-1-1) to (t-2-1-3);
		
		\draw [red,line width=1.3pt,->] (s-2-1-1) to (u-2-1-3);
		\draw [red,line width=1.3pt,->] (u-2-1-4) to (v-2-1-2);
		\draw [red,line width=1.3pt,->] (v-2-1-3) to (t-2-1-2);
		
		\draw [blue,line width=1.3pt,->] (s-2-1-2) to (u-2-1-2);
		\draw [blue,line width=1.3pt,->] (u-2-1-1) to (w-2-1-6);
		\draw [blue,line width=1.3pt,->,dashed] (w-2-1-5) to (v-2-1-1);
		\draw [blue,line width=1.3pt,->] (v-2-1-4) to (t-2-1-1);
		
		\draw [red,line width=1.3pt,->,dashed] (s-2-1-3) to (w-2-1-3);
		\draw [red,line width=1.3pt,->] (w-2-1-4) to (t-2-1-4);
		
		\node (l-sw-2-1) [position=130:17.5mm from c-2-1] {$\frac{0}{1}$};
		\node (l-wt-2-1) [position=50:17.5mm from c-2-1] {$\frac{1}{2}$};
		\node (l-uw-2-1) [position=160:7mm from c-2-1] {$\frac{1}{1}$};
		\node (l-wv-2-1) [position=20:7mm from c-2-1] {$\frac{0}{1}$};
		\node (l-uv-2-1) [position=270:12mm from c-2-1] {$\frac{1}{1}$};
		\node (l-su-2-1) [position=200:15mm from c-2-1] {$\frac{2}{2}$};
		\node (l-vt-2-1) [position=340:15mm from c-2-1] {$\frac{1}{1}$};

		\draw [cyan,line width=6.5pt,line cap=round,opacity=0.2] (s-2-2) to (u-2-2);
		\draw [cyan,line width=6.5pt,line cap=round,opacity=0.2] (u-2-2) to (w-2-2);
		\draw [cyan,line width=6.5pt,line cap=round,opacity=0.2] (w-2-2) to (t-2-2);
		
		\draw [Magenta,line width=6.5pt,line cap=round,opacity=0.2] (s-2-2) to (w-2-2);
		\draw [Magenta,line width=6.5pt,line cap=round,opacity=0.2] (w-2-2) to (t-2-2);
		
		\draw [blue,line width=0.9pt,->,dashed] (s-2-2-4) to (w-2-2-2);
		\draw [blue,line width=1.3pt,->] (w-2-2-1) to (t-2-2-3);
		
		\draw [red,line width=0.9pt,->,dashed] (s-2-2-1) to (u-2-2-3);
		\draw [red,line width=1.3pt,->] (u-2-2-4) to (v-2-2-2);
		\draw [red,line width=1.3pt,->] (v-2-2-3) to (t-2-2-2);
		
		\draw [blue,line width=1.3pt,->] (s-2-2-2) to (u-2-2-2);
		\draw [blue,line width=1.3pt,->] (u-2-2-1) to (w-2-2-6);
		\draw [blue,line width=1.3pt,->,dashed] (w-2-2-5) to (v-2-2-1);
		\draw [blue,line width=1.3pt,->] (v-2-2-4) to (t-2-2-1);
		
		\draw [red,line width=1.3pt,->] (s-2-2-3) to (w-2-2-3);
		\draw [red,line width=1.3pt,->] (w-2-2-4) to (t-2-2-4);
		
		\node (l-sw-2-2) [position=130:17.5mm from c-2-2] {$\frac{1}{1}$};
		\node (l-wt-2-2) [position=50:17.5mm from c-2-2] {$\frac{2}{2}$};
		\node (l-uw-2-2) [position=160:7mm from c-2-2] {$\frac{1}{1}$};
		\node (l-wv-2-2) [position=20:7mm from c-2-2] {$\frac{0}{1}$};
		\node (l-uv-2-2) [position=270:12mm from c-2-2] {$\frac{0}{1}$};
		\node (l-su-2-2) [position=200:15mm from c-2-2] {$\frac{1}{2}$};
		\node (l-vt-2-2) [position=340:15mm from c-2-2] {$\frac{0}{1}$};

		\draw [cyan,line width=6.5pt,line cap=round,opacity=0.2] (s-2-3) to (u-2-3);
		\draw [cyan,line width=6.5pt,line cap=round,opacity=0.2] (u-2-3) to (w-2-3);
		\draw [cyan,line width=6.5pt,line cap=round,opacity=0.2] (w-2-3) to (v-2-3);
		\draw [cyan,line width=6.5pt,line cap=round,opacity=0.2] (v-2-3) to (t-2-3);
		
		\draw [Magenta,line width=6.5pt,line cap=round,opacity=0.2] (s-2-3) to (w-2-3);
		\draw [Magenta,line width=6.5pt,line cap=round,opacity=0.2] (w-2-3) to (t-2-3);
		
		\draw [blue,line width=0.9pt,->,dashed] (s-2-3-4) to (w-2-3-2);
		\draw [blue,line width=0.9pt,->,dashed] (w-2-3-1) to (t-2-3-3);
		
		\draw [red,line width=0.9pt,->,dashed] (s-2-3-1) to (u-2-3-3);
		\draw [red,line width=0.9pt,->,dashed] (u-2-3-4) to (v-2-3-2);
		\draw [red,line width=0.9pt,->,dashed] (v-2-3-3) to (t-2-3-2);
		
		\draw [blue,line width=1.3pt,->] (s-2-3-2) to (u-2-3-2);
		\draw [blue,line width=1.3pt,->] (u-2-3-1) to (w-2-3-6);
		\draw [blue,line width=1.3pt,->] (w-2-3-5) to (v-2-3-1);
		\draw [blue,line width=1.3pt,->] (v-2-3-4) to (t-2-3-1);
		
		\draw [red,line width=1.3pt,->] (s-2-3-3) to (w-2-3-3);
		\draw [red,line width=1.3pt,->] (w-2-3-4) to (t-2-3-4);
		
		\node (l-sw-2-3) [position=130:17.5mm from c-2-3] {$\frac{1}{1}$};
		\node (l-wt-2-3) [position=50:17.5mm from c-2-3] {$\frac{1}{2}$};
		\node (l-uw-2-3) [position=160:7mm from c-2-3] {$\frac{1}{1}$};
		\node (l-wv-2-3) [position=20:7mm from c-2-3] {$\frac{1}{1}$};
		\node (l-uv-2-3) [position=270:12mm from c-2-3] {$\frac{0}{1}$};
		\node (l-su-2-3) [position=200:15mm from c-2-3] {$\frac{1}{2}$};
		\node (l-vt-2-3) [position=340:15mm from c-2-3] {$\frac{1}{1}$};
		\end{pgfonlayer}

		\end{tikzpicture}	
	\end{center}
	\vspace{-1em}
	\caption{\emph{Example} for \Cref{alg:main}. The $2$ update flow pairs are \textcolor{red}{red} and 
		\textcolor{blue}{blue},
		each of demand $1$. 
		The active edges of the respective colors are indicated as \emph{solid lines} and the inactive edges are \emph{dashed}. 
		Each edge in the flow graph is annotated with its current
		load (\emph{top}) and its capacity (\emph{bottom}).
		We start by identifying the \textcolor{blue}{blue} and \textcolor{red}{red} blocks. For \textcolor{red}{red} there is exactly one such block \textcolor{red}{$r_1$}, since \textcolor{red}{$R^o$} and \textcolor{red}{$R^u$} only coincide in $s$ and $t$. The \textcolor{blue}{blue} flow pair on the other hand omits two blocks \textcolor{blue}{$b_1$} and \textcolor{blue}{$b_2$}: \textcolor{blue}{$B^o$} and \textcolor{blue}{$B^u$} meet again at $w$ and at $t$.
		We observe that $\textcolor{blue}{b_2}$
		can only be updated after $\textcolor{red}{r_1}$ 
		has been updated; similarly, $\textcolor{red}{r_1}$
		can only be updated after $\textcolor{blue}{b_1}$
		has been updated.
		An update sequence respecting these dependencies
		can be constructed as follows.
		We can first prepare the blocks by updating
		the following two out-edges which currently
		do not carry any flow:
		$(w,\textcolor{red}{red})$, 
		$(u,\textcolor{blue}{blue})$, and $(v,\textcolor{blue}{blue})$.
		Subsequently, the three blocks can be updated in a congestion-free
		manner in the following order:
		Prepare the update for all blocks in the first round. Then, update
		$\textcolor{blue}{b_1}$ in the second round,
		$\textcolor{red}{r_1}$ in the third round,
		$\textcolor{blue}{b_2}$ in the fourth round.
	} 
	\label{fig:exwithblocks}
\end{figure*}
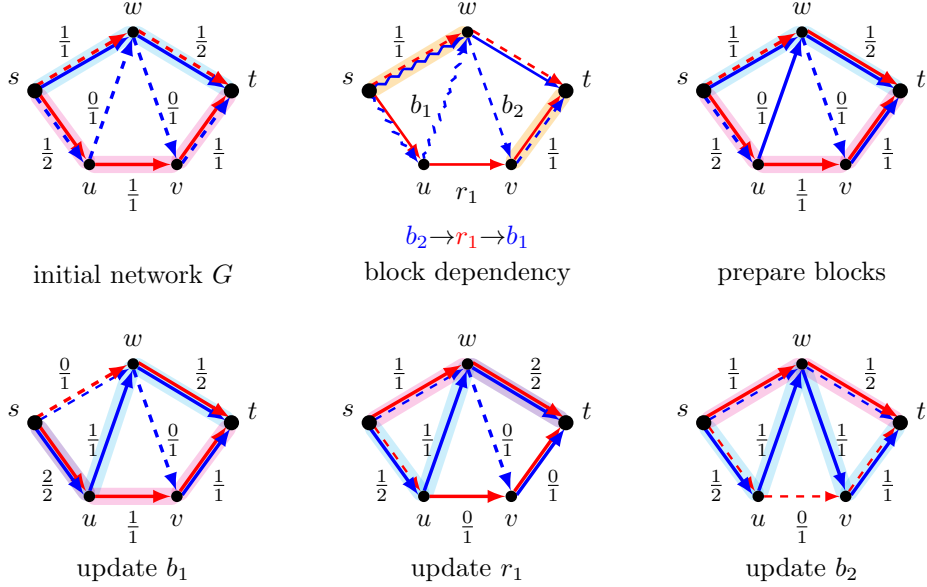

Suppose $\updatesequence$ is a feasible update sequence for $G$. We say 
that a $c$-block $b$
w.r.t.~$\updatesequence=(\round_1,\ldots,\round_\ell)$ is \emph{updated in
	consecutive rounds}, if the following holds:
if some of the edges of $b$ 
are activated/deactivated in
round $i$ and some others in round $j$, 
then for every $i<k<j$, there is
an edge of $b$ which is activated/deactivated. 

\begin{lemma}\label[lemma]{lem:updateblockstart}
	Let $b$ be a $c$-block. Then in
	a feasible update sequence $\updatesequence$, all vertices
	(resp.~their outgoing $c$ flow edges)
	in $F^u_c\cap b - \Start{b}$ are
	updated  strictly before $\Start{b}$. Moreover, all vertices in $b-F^u_c$ are updated strictly
	after $\Start{b}$ is updated.
\end{lemma}

\begin{proof}
	In the following, we will implicitly assume
	flow $c$, and will not mention 
	it explicitly everywhere.
	We will write $F^u_b$ for $F^u_c\cap b$ and
	$F^o_b$
	for $F^o_c\cap b$.
	For the sake of contradiction, let $U=\{v\in V(G)\mid v\in F^u_b-F^o_b-\Start{b},\updatesequence(v,c) > \updatesequence(\Start{b},c)\}$.
	Moreover, let $v$ be the
	vertex of $U$ which is updated the latest
	and $\updatesequence(v,c) = \max_{u\in U}\updatesequence(u,c)$.
	By our condition, the update of $v$ enables 
	a transient flow along
	edges in $F^u_c\cap b$. Hence, 
	there now exists an $(s,t)$-flow through $b$ using only update edges.
	
	No vertex in $F_1\coloneqq F^o_b-(F^u_b-\Start{b})$ 
	could have been updated before, or
	simultaneously with $v$: 
	otherwise, between the time $u$ has been updated
	and before the update
	of $v$, there would not exist a transient flow.
	But once we update $v$, there is a $c$-flow which 
	traverses the
	vertices in $F^o_b-F^u_b$, and another $c$-flow which 
	traverses
	$v\not\in F_1$: a contradiction. Note that $F_1\neq \emptyset$.
	The other direction is obvious: updating
	any vertex in $(F^o_c\cap b)-F^u_c$ inhibits any transient flow.
\end{proof}

\begin{lemma}\label[lemma]{lem:updatewholeblock}
	Given any feasible (not necessarily shortest)
	update sequence $\updatesequence$, there is a feasible 
	update sequence $\updatesequence'$  which updates every block in at most 
	$3$ consecutive rounds.
\end{lemma}
\begin{proof}
Consider the following approach to 
update free blocks
(for the $i$-block $b$): 
first resolve $(v,P_i)$ for all $v\in P^u_i \cap b-\Start{b}$;
then resolve $(\Start{b},P_i)$; finally
resolve $(v,P_i)$ for all $v\in (b-P^u_i)$.

Now let $\updatesequence$ be a feasible update sequence with a minimum number of 
blocks which are not updated in~$3$ consecutive rounds. 
Furthermore let $b$ be such a $c$-block. Let $i$ 
be the round in which $\Start{b}$ is updated. Then
by~\Cref{lem:updateblockstart},
all other vertices of
$F^u_c\cap b$ have been updated 
in the previous rounds. Moreover, since they do not carry
any flow during these rounds, the edges can all be updated in round $i-1$. 
By our assumption, we can 
update $\Start{b}$ in round $i$, and hence now 
this is still possible.
	
As $\Start{b}$ is updated in round~$i$, the edges of 
$F^o_c\cap b$ do not carry any active $c$-flow in 
round~$i+1$ and thus we can deactivate all remaining 
such edges in this round. This is a contradiction to the 
choice of $\updatesequence$, and hence there is always a feasible 
sequence $\updatesequence$ satisfying the requirements of the 
lemma. 
	
In particular, the above algorithm is correct.
\end{proof}

From the above lemmas, we immediately 
derive a corollary regarding the optimality in terms 
of the number of rounds: 
the~$3$ rounds feasible update sequence.

\begin{corollary}\label{cor:blockroundoptimal}
	Let $b$ be any $c$-block with 
	$\left| E(b\cap F_c^o)\right|\geq 2$ and 
	$\left| E(b\cap F_c^u)\right|\geq 2$. Then it is 
	not possible to update $b$ in less than~$3$ 
	rounds: otherwise it is not possible to 
	update~$b$ in less than~$2$ rounds.	
\end{corollary}

Next we show that if there is a cycle in the 
dependency graph, then
it is impossible to update any flow.

\begin{lemma}\label[lemma]{lem:nocycle}
	If there is a cycle in the dependency graph, then 
	there is no feasible update sequence.
\end{lemma}
\begin{proof} 
Suppose that there exists a cycle in the dependency graph.
Without loss of generality, we can assume that this is the only
cycle in the dependency graph as
we can always remove vertices without creating new dependencies. 
Then it is not possible to update the cycle.
For the sake of contradiction, suppose 
that there is a feasible update order; then there is a feasible 
update order in which blocks are updated in consecutive (distinct) 
rounds. But in this order, one of the vertices in 
the dependency graph (a block) 
should be earlier than the others. 
This is impossible due to dependency 
on other vertices. 
\end{proof}

We will now slightly modify \Cref{alg:main} to create a new algorithm 
which not only computes 
a feasible sequence $\updatesequence$ for a given update flow network in polynomial time,
whenever it exists, but which also ensures that $\updatesequence$ is as short as possible
(in terms of number of rounds). For any block $b$, let $c(b)$ denote its corresponding flow pair.

\begin{algorithm}\textbf{Optimal $2$-Flow DAG Update}\label[algorithm]{alg:mainopt} 
	\begin{enumerate}
		\item [] \textbf{Input: Update Flow Network $G$}
		\item Compute the dependency graph $D$ of $G$. \label{line:1}
		\item If there is a cycle in $D$, return \emph{impossible to update}.
		\item If there is any block $b$ corresponding to a
		sink vertex of $D$ with $(b\cap
		F^u_{c(b)})-\Start{b}\neq\emptyset$ set
		$i\colon\!\!\!=2$, otherwise set $i\colon\!\!\!=1.$
		\item While $D\neq \emptyset$ repeat:
		\begin{enumerate}[i]
			\item \label{lbl:stepopt} Schedule the update of all blocks $b$ which correspond to the sink 
			vertices of $D$
			 for the
			rounds $i-1$, $i$, $i+1$, such that
			$\Start{b}$ is updated in round $i$.
			\item Delete all of the current sink vertices from $D$.
			\item Set $i\colon\!\!\!= i+1$.
		\end{enumerate}
	\end{enumerate}
\end{algorithm}

\begin{theorem}
	An optimal (feasible) update sequence 
	on acyclic update flow networks with exactly~$2$ update flow pairs can be found in linear time.
\end{theorem}
\begin{proof}
	Let $G$ denote the given update flow network. 
	In the following, for ease of presentation, we will slightly abuse 
	terminology and say that ``a block is updated in some round'',
	meaning that the block is updated in the corresponding consecutive rounds 
	as in the proof of
	\Cref{lem:updatewholeblock}.
	
	We proceed as follows. First, we find a block decomposition and 
	create the dependency graph of the input instance.
	This takes linear time only. 
	If there is a cycle in that graph, we output \emph{impossible}
	(cf~\Cref{lem:nocycle}). 
	Otherwise, we apply \Cref{alg:mainopt}. 
	As there is no cycle in the dependency
	graph (a property that stays invariant), 
	in each round, either there exists a free block which is not
	processed yet, or everything is already updated or is in the process of being updated. Hence, if there is a
	feasible solution (it may not be unique), we can find one in time~$O(\sizeof{G})$. 
	
	For the optimality in terms of the number of rounds, 
	consider two feasible update sequences. Let $\updatesequence_{\text{\sc Alg}}$ be the update
	sequence produced by \Cref{alg:mainopt} and let $\updatesequence_{\text{\sc
			Opt}}$ be a feasible update sequence that realizes the minimum
	number of rounds. According to~\Cref{lem:updateblockstart},
	any block $b$ is processed only in 
	round $\Start{b}$.
	
	Suppose there is a block $b'$ such that $\round_{\text{\sc
			Opt}}(b')<\round_{\text{\sc Alg}}(b')$. Then let $b$ be the block
	with the smallest such $\round_{\text{\sc Opt}}(b)$. Hence, for every
	block $b''$ with $\round_{\text{\sc Opt}}(b'')\leq\round_{\text{\sc
			Opt}}(b)$, $\round_{\text{\sc Opt}}(b'')\geq\round_{\text{\sc
			Alg}}(b'')$ holds. Since $\Start{b}$ is updated in round
	$\round_{\text{\sc Opt}}(b)$, there are no dependencies for $b$ that are
	still in place in this round. Thus, 
	according to the
	sequence $\updatesequence_{\text{\sc Opt}}$,
	$b$ is a sink vertex of the
	dependency graph after round $\round_{\text{\sc Opt}}(b)-1$.
	Furthermore, by our previous
	observation, every start of some block has been updated up to this
	round in the optimal sequence, and hence it is also already updated in the same
	round in $\updatesequence_{\text{\sc Alg}}$. This means that after round
	$\round_{\text{\sc Opt}}(b)-1<\round_{\text{\sc Alg}}(b)-1$, $b$ is a sink
	vertex of the dependency graph of $\updatesequence_{\text{\sc Alg}}$ as well.
	Thus, \Cref{alg:mainopt} would have scheduled the update of
	block $b$ in the rounds $\round_{\text{\sc Opt}}(b)-1$, $\round_{\text{\sc
			Opt}}(b)$ and $\round_{\text{\sc Opt}}(b)+1$. Contradiction.
	
	Thus $\round_{\text{\sc Alg}}(b)\leq \round_{\text{\sc Opt}}(b)$ 
	for all blocks $b$. Now let $b_1,\dots,b_{\ell}$ be the last 
	blocks whose starts are updated the latest under 
	$\updatesequence_{\text{\sc Alg}}$. If there is some $i\in[\ell]$ 
	such that $\left| E_{b_i}^o\right|\geq 2$ and $\left| 
	E_{b_i}^u\right|\geq 2$, $\updatesequence_{\text{\sc Alg}}$ uses 
	exactly $\round_{\text{\sc Alg}}(b_i)+1$ rounds; 
	otherwise it is one round less, by Corollary~\ref{cor:blockroundoptimal}. 
	By our previous observation, none of these blocks can start later than 
	$\round_{\text{\sc Alg}}(b_i)$ and thus $\round_{\text{\sc Opt}}$ uses at 
	least as many rounds as \Cref{alg:mainopt}. Hence the algorithm is 
	optimal in the number of rounds. 
\end{proof}

\section{NP-hardness for More Flows}\label{sec:np-hard}

This section shows that the polynomial-time result derived
above cannot be generalized much further: 
it is NP-hard to compute a shortest schedule already for six flows,
and even if the pair of old and new path \emph{forms a DAG}.
In fact, we show that already the decision problem, i.e., whether a feasible
schedule exists, is NP-hard.

\begin{theorem}
\label{thm:6flow_hardness}
Deciding whether a feasible network update schedule exists for 
a given update flow network in which each flow pair 
forms a DAG is $NP$-hard for 6 flows.
\end{theorem}

We use a reduction from 3-SAT. Let $C$ be any 3-SAT formula 
with $n$ variables $x_1,\dots, x_n$ and 
$m$ clauses $C_1,\dots, C_m$. The resulting update flow 
network is denoted as $G(C)$.

We will create $6$ flow pairs: $X$, $\overline{X}$, $D_1$, 
$D_2$, $D_3$ and $B$, each having demand $1$. $B$is
the blocking
pair: it can by updated only if all clauses are satisfied. 
Flows $X$ and $\overline{X}$ contain gadgets
for all literals, $X$ for positive ones and $\overline{X}$ for negative
ones. Updating a variable gadget in $X$
corresponds to assigning the variable value $1$ in $C$. 
Flow $B$ prevents the variable gadget to be updated
in both $X$ and $\overline{X}$, unless all clauses are satisfied.

Flows $D_1$, $D_2$ and $D_3$ encode clauses of $C$. 
Each of these flows contains a clause gadget linking 
a clause to one of its literals. This gadget can be updated only if the literal is 
satisfied. Updating a clause gadget
in one of those flows will allow $B$ to be updated.

Now we proceed with the detailed description of the reduction.

\begin{enumerate}
	\item \textbf{Clause gadgets:} For every clause $i\in[m]$ we 
	introduce eight vertices: $u^i$, $v^i$
	and for $j\in \{1,2,3\}$ $u^i_j$ and $v^i_j$. For $j\in \{1,2,3\}$ 
	we add edge $(u^i, v^i)$ to $D^o_j$
	and edges $(u^i, u^i_j)$, $(u^i_j, v^i_j)$ and $(v^i_j, v^i)$ to $D^u_j$.
	\item \textbf{Variable Gadgets:} For every $j\in[n]$, we introduce two
		vertices: $w_1^j$ and $w_2^j$. Let $P_j=\left\{
		p_1^j,\dots,p_{k_j}^j \right\}$ denote the set of indices of the
		clauses containing the literal $x_j$ and $\overline{P}_j=\left\{
		\overline{p}^j_1,\dots,\overline{p}^j_{k'_j} \right\}$ the set of
		indices of the clauses containing the literal
		$\overline{x}_j$. Furthermore, let $\pi(i,j)$ denote the position of
		$x_j$ in the clause $C_i$, $i\in P_j$. Similarly,
		$\overline{\pi}(i',j)$ denotes the position of $\overline{x_j}$ in $C_{i'}$ where 
		$i'\in\overline{P}_j$.

		To $X^u$ and $\overline{X}^u$ we add edge $(w_1^j, w_2^j)$.

		To $X^o$ we add the following edges:
		\begin{itemize}
			\item $(u^{p_i^j}_{\pi(p_i^j,j)}, v^{p_i^j}_{\pi(p_i^j,j)})$ for all $i\in \{1,\dots,k_j\}$,
			\item $(v^{p_i^j}_{\pi(p_i^j,j)}, u^{p_{i+1}^j}_{\pi(p_{i+1}^j,j)})$ for all $i\in \{1,\dots,k_j-1\}$,
			\item $(w_1^j, u^{p_{1}^j}_{\pi(p_{1}^j,j)})$ and $(v^{p_k^j}_{\pi(p_k^j,j)},w_2^j)$.
		\end{itemize}

		We proceed similarly with $\overline{X}^o$ and 
		clauses containing $\overline{x}_j$.
	\item \textbf{Blocking flow:} The goal of flow $B$ is to 
		block update of $w_1^j$, for any $j\in[n]$, in both $X$ and $\overline{X}$.

		To do that we add to $B^o$ the following edges:
		\begin{itemize}
			\item $(w_1^j, w_2^j)$, for all $j\in[n]$,
			\item $(w_2^j, w_1^{j+1})$, for all $j\in[n-1]$.
		\end{itemize}

		We also add the following edges to $B^u$:
		\begin{itemize}
			\item $(u^i, v^i)$, for all $i\in[m]$,
			\item $(v^i, u^{i+1})$, for all $i\in[m-1]$.
		\end{itemize}
	\item \textbf{Source and Terminal:} Now we need to 
	connect all the gadgets in the flows.
		The source and the terminal of all flows will be $s$ and $t$.

		To $D_j^o$ and $D_j^u$, for $j\in\{1,2,3\}$, we add the following edges:
		\begin{itemize}
			\item $(v^i,u^{i+1})$, for all $i\in[m-1]$,
			\item $(s, u^1)$ and $(v^m, t)$.
		\end{itemize}

		To $X^o$, $X^u$, $\overline{X}^o$ and $\overline{X}^u$ 
		we add the following edges:
		\begin{itemize}
			\item $(w^j_2, w^{j+1}_1)$, for all $j\in[n-1]$
			\item $(s, w^1_1)$ and $(w^n_2, t)$
		\end{itemize}

		We also add edges $(s, w^1_1)$ and $(w^n_2, t)$ to $B^o$ and 
		edges $(s, u^1)$ and $(v^m, t)$ to $B^u$.
 		
 	\item \textbf{Edges capacity:} For all $j\in[n]$ we set the capacity of
 		edge $(w^j_1, w^j_2)$ to be $2$.
 		Also for all $i\in[m]$ we set capacity of edge $(u^i,v^i)$ to be $3$
 		and capacity of edge $(u^i_j,v^i_j)$, for $j\in\{1,2,3\}$, to be $1$.

 		For all the other edges, we set their capacity to be $6$, that is, to the
 		number of flows. Therefore they cannot violate any capacity constraint.
\end{enumerate}

\begin{lemma} Given any valid update sequence $\updatesequence$ 
		for the above constructed update flow network $G(C)$, 
		the following conditions hold.
		\begin{enumerate}
			\item For every $r<\updatesequence(s,B)$ and $j\in[n]$ $\updatesequence(w^j_1, X) > r$ or $\updatesequence(w^j_1, \overline{X}) > r$.
			\label{con:nphardvar}
			\item For every $r\geq\updatesequence(s,B)$ and $i\in[m]$ $\updatesequence(u^i,D_1) < r$, $\updatesequence(u^i,D_2) < r$ or $\updatesequence(u^i,D_3) < r$.
			\label{con:nphardclause}
		\end{enumerate}
		\label{lem:6flowhardaux}
	\end{lemma}
	\begin{proof}
		Note that $B^o$ and $B^u$ have no common nodes apart from
		$s$ and $t$. That means that for any $r$ either $T_{B,U_r} = B^o$
		or $T_{B,U_r} = B^u$. Now we prove both conditions.
		\begin{enumerate}
			\item Let us consider any $j\in[n]$. 
			As $r<\updatesequence(s,B)$, then $T_{B,U_r} = B^o$. 
			The capacity of edge $(w^j_1, w^j_2)$ is $2$ and it belongs 
			to $B^o$. Therefore it can be in at most
			one other transient flow, so the condition holds.
			\item Let us consider any $i\in[m]$. As $r\geq\updatesequence(s,B)$, 
			then $T_{B,U_r} = B^u$.
			The capacity of edge $(u^i, v^i)$ is $3$ and it belongs to $B^u$.
			Therefore it can be in at most two other transient flows, 
			so the condition holds.
		\end{enumerate}
	\end{proof}

	\begin{proof}

	Now we are ready to prove Theorem \ref{thm:6flow_hardness}. 
	First, let us assume that $C$ is satisfiable and
	we will construct valid update sequence for $G(C)$.
	Let $\sigma$ be an assignment satisfying $C$. Then the
	update sequence for $G(C)$ is as follows.
	\begin{enumerate}
		\item For every $j\in[n]$, if $\sigma(x_j) = 1$ then resolve $(w^j_1,X)$,
			otherwise resolve $(w^j_1,\overline{X})$.
		\item For every clause $C_i$ at least one of 
		the edges $(u^i_1,v^i_1)$, $(u^i_2,v^i_2)$ 
			and $(u^i_3,v^i_3)$ is neither in $T_{X,r_2^f-1}$ nor $T_{\overline{X},r_2^f-1}$. So the update of $u^i$ can be resolved in 
			the corresponding flow $D_1$, $D_2$ or $D_3$ (this follows
			from $\sigma$ being satisfying assignment).
		\item As every $i\in[m]$ edge $(u^i,v^i)$ is used 
		by at most $2$ flows,
			we can resolve every update in $B^u$, resolving $(s,B)$ as the last one.
		\item For every $j\in[n]$, resolve either $(w^j_1,X)$ or $(w^j_1,\overline{X})$,
			depending on which one was not resolved in step 1.
		\item For every $i\in[m]$ resolve updates of $u^i$ 
		in flows $D_1$, $D_2$ and $D_3$
			(those that have not been resolved in step 2).
		\item Resolve the remaining updates in all flows.

	\end{enumerate}

	Now let us assume that there is a valid update 
	sequence $\sigma$ for $G(C)$.
	We will show that $C$ is satisfiable by 
	constructing satisfying assignment $\sigma$.

	Let us consider round $r=\sigma(s,B)$. We assign values
	 in the following way. For $j\in[n]$, if 
	 $\sigma(w_1,X) < r$ then $\sigma(x_j) := 1$
	 and if $\sigma(w_1,\overline{X}) < r$ then $\sigma(x_j) := 0$. If both 
	 $\sigma(w_1,X) > r$ and $\sigma(w_1,\overline{X}) > r$ 
	 we assign to $x_j$ a random
	 value. By Condition \ref{con:nphardvar} of Lemma \ref{lem:6flowhardaux} 
	 this is a correct assignment,
	 that is no variable is assigned two values.

	 We want to prove that this assignment satisfies $\sigma$. 
	 Let us consider any clause $C_i$.
	 By Condition \ref{con:nphardclause} of Lemma \ref{lem:6flowhardaux} 
	 at least one of $(u^i,D_1)$, $(u_i,D_2)$ or $(u_i, D_3)$
	 is updated before round $\roundr$. That m0eans that at least for one of variables $x_j$
	 in $C_i$ $\sigma(w^j_1,X) < r$, if $C_i$ contains literal $x_j$, or $\sigma(w^j_1,\overline{X}) < \roundr$, if $C_i$
	 contains literal $\overline(x)_j$. This means that
	 $C_i$ is satisfied by $x_j$ in $\sigma$.

	 \end{proof}

\section{Optimal Scheduling of Arbitrary Problems}\label{sec:mip}

Since the problem is generally NP-hard, for completeness and in order to
investigate the runtime of such an approach, we in the following
describe an optimal scheduling algorithm for a general model
and arbitrary number of flows, which runs in super-polynomial time.
The algorithm is based on mixed integer linear programming.

The formulation first reserves variables for all possible rounds during which
a node can update a flow (Constraints (\ref{line:1})).
Henceforth, we refer to them as \textit{schedule variables}.
Schedule variables are constrained so that a node updates each of its flows only once.
The remaining constraints ensure the following feasibility criteria:
\begin{enumerate}
	\item
	Constraints (\ref{LP:init_oldedges}) to (\ref{LP:forknode}) prepare the variables
	used in the consistency checks.
	\item 
	Assuming a value assignment to the schedule variables,
	Constraints (\ref{LP:transientgraph1}) to (\ref{LP:transientflow_check}) emulate the update
	with respect to the schedule and ensure no flow is interrupted during any update round $r$.
	That is, for any link that receives a new flow during $r$, the incident node
	must have been already updated in an earlier round ($<r$).
	Also, any node that removes a flow from an outgoing link
	must postpone this to a later round ($>r$). 
	However, there is an exception for nodes that are incident
	to both old and new flow links (denoted by fork nodes).
	This criteria accounts for the fact that node updates occur
	asynchronously and the flows must not be interrupted  in any case.
	\item 
	With the last set of constraints, we ensure that during the emulation all capacities are respected.

\end{enumerate} 

\begin{figure}[!h]
\begin{IEEEeqnarray}{ll}
	\textbf{Minimize}\;R  \label{LP:objective}
	\\
	  \textit{\small{ROUNDS}} = \{1,..,(|V|-1).|P|\}    \nonumber
	\\
	 \sum_{r\in \textit{\scriptsize{ROUNDS}}} x^r_{v,i} = 1	&	\forall i\in |P|, v\in P_i \label{LP:schedule}		\IEEEeqnarraynumspace
	\\
	y^0_{(u,v),i} = 1	&	\forall (u,v)\in F_i^o \label{LP:init_oldedges}
	\\
	y^0_{(u,v),i} = 0	&	\forall (u,v)\in F_i^u \label{LP:init_newedges}
	\\
	 \forall i\in [|P|], r\in \textit{\small{ROUNDS}} \;
	 	\boldsymbol{\{}	&  \label{LP:repeat}
	\\
	 x^r_{v,i},\textit{fork}^r_{v,i},\textit{join}^r_{v,i} \in \{0,1\}   &	\forall v\in P_i
	\\
	 y^r_{(u,v),i},f^r_{(u,v),i} \in [0,1]	
	 &	\forall (u,v)\in P_i\label{LP:transientgraph1}
	\\
	 R \geq r.x^r_{v,i}	&	\forall v\in P_i \label{LP:rounds}
	\\
	y^r_{(u,v),i} = \sum_{r'\leq r} x^{r'}_{u,i} &	\forall (u,v)\in F_i^u \label{LP:activeflag1}
	\\
	y^r_{(u,v),i} = 1 - \sum_{r'\leq r} x^{r'}_{u,i}   &	\forall (u,v)\in F_i^o \label{LP:activeflag2}
	\\
	\textit{fork}^r_{v,i} =		
		\begin{cases}
			x^r_{v,i} &	\hspace{-0.5em} \exists w,w'\in P_i:
					\begin{cases}
						(v,w)\in F_i^o \\
						(v,w')\in F_i^u
					\end{cases}	\\
			0	& \text{else}
		\end{cases}
		&	\forall v\in P_i \label{LP:forknode}
	\\
	f^r_{(u,v),i} \leq y^{r-1}_{(u,v),i} + \textit{fork}^r_{u,i}
		&	\forall (u,v)\in P_i\label{LP:transientgraph1}
	\\
	f^r_{(u,v),i} \leq y^{r}_{(u,v),i} + \textit{fork}^r_{u,i}
		&	\forall (u,v)\in P_i \label{LP:transientgraph2}
	\\
	\textit{join}^r_{v,i} \leq		
		f^{r}_{(u,v),i},f^{r}_{(u',v),i}	
		&	\forall v,u,u' \! \in \! P_i:\!
				\begin{cases}
						(u,v) \! \in \! F_i^o	\\
						(u',v) \! \in \! F_i^u
				\end{cases}	\label{LP:joinnode}
	\\
	\sum_{(v,w)\in P_i} f^{r}_{(v,w),i} - \sum_{(u,v)\in P_i} f^{r}_{(u,v),i} =		
	\begin{cases}
					1 + \textit{fork}^r_{s,i}  &  v=s\\
					-(1 + \textit{join}^r_{t,i}) &  v=t\\
					\textit{fork}^r_{v,i} - \textit{join}^r_{v,i}  &  \text{else}
			\end{cases}
		&	\forall v\in P_i \label{LP:transientflow_check}
	\\
	 \boldsymbol{\}} \nonumber
	\\
	\sum_{i\in [|P|]} f^r_{(u,v),i} \leq C_{(u,v)}	&	\hspace{-5em}	\forall r \in \textit{\small{ROUNDS}}, (u,v)\in E \label{LP:capacitycheck}
\end{IEEEeqnarray}
\caption{Mixed Integer Program for $k$ flow pairs}
\label{MIP}
\end{figure}
	
Next, we describe the formulation in detail.
\begin{itemize}
	\item (\ref{LP:schedule}): Each schedule variable $x^r_{v,i}$ indicates whether a node $v$ is scheduled to update flow $i$ in round $\roundr$.
	
	\item (\ref{LP:repeat}): Repeat the embraced lines for every pairs $P_i$ and each  round $\roundr \in \textit{ROUNDS}$.
	
	\item (\ref{LP:activeflag1}),(\ref{LP:activeflag2}): $y^r_{(u,v),i}$ indicates whether the link $(u,v)$ is active for pair $i$ immediately after round $\roundr$ (i.e. active graph).
	
	\item (\ref{LP:forknode}): \textit{fork nodes} are the nodes at which old and update paths split.
	A fork node $v$ acts as a source, doubling its incoming transient flow $i$, when it updates the flow during round $\roundr$ (i.e.~if $\textit{fork}^r_{v,i}$ is 1).
	
	\item (\ref{LP:joinnode}): \textit{join nodes} are nodes at which the old and update paths meet once again.
	A join node acts as a sink (if $\textit{join}^r_{v,i}$ is 1)  when the two in-links both carry the transient flow $i$ in the transient state of round $\roundr$.
	
	\item (\ref{LP:transientgraph1}),(\ref{LP:transientgraph2}): $f^r_{u,v,i}$ specifies the transient flow $i$ on a link $(u,v)$.
	The first terms on the r.h.s.~constrains together state that the link is allowed to be utilized in the transient state of round $\roundr$,
	 if it is active before round $\roundr$ and it remains active during the round.
	Alternatively,  if the link is deactivating in round $\roundr$ due to the updating fork node $u$, then the second term allows the link to be usable in the transient state.
	 (This, along with Constraint (\ref{LP:transientflow_check}) guarantees there will be no loops on the old out-branch of any updating fork node.)

	\item (\ref{LP:transientflow_check}): Runs a variable-size transient flow $i$ from $s$ to $t$
	in order to impose $st$-connectivity in the worst-case transient state.
	The flow produced at $s$ is of size 1 and it arrives at $t$ with the same size.
	In the meanwhile, any active fork node (including possibly $s$) adds one unit to this flow and
	splits it into two unit-size flows along both its out-links.
	Later, a join node consumes this extra flow by taking away the 1 unit.
	
	\item (\ref{LP:capacitycheck}): The capacity constraints.
\end{itemize}

Because of a possible cleanup round after a fork node updates,
it is necessary to maintain $st$-connectivity via both (old and update) out-links of the fork node,
which is ensured by (15). In other words, no cleanup (i.e.~removal of old flow rules) should occur
on the old branch of the fork node in the same round it reroutes to the new branch.

\section{Empirical Results}\label{sec:sims}

In order to gain insights into the actual number of rounds
needed to reroute flows in real networks, we conducted
a simulation study on real network topologies. In particular,
we want to compare the length of the schedules produced by our
algorithm (which provably provides \emph{shortest}
schedules) to the state-of-the-art algorithm presented in~\cite{icalp18}
(which only computes \emph{feasible} schedules).
In order to study the need for fast algorithms, we
compare the runtime of our algorithm to the mathematical
programming approach, as it is frequently used in the literature~\cite{update-survey}.

We implemented  \Cref{alg:mainopt} using standard C++ libraries and performed
a exhaustive evaluation on over 100  topologies provided by
\cite{topologyZoo} and $\approx 136$ million records in total.
The graphs were chosen so that the runtime would be practical.
In another, but similar implementation we evaluated \Cref{alg:main}
on the same input cases and
compared the number of rounds obtained from these algorithms in \Cref{fig:charts_rounds}.


We observe that our Algorithm~\ref{alg:main} is a simpler
feasibility algorithm than~\cite{icalp18}, for two flows:
it employs basic batching, which leads to shorter schedules
compared to ~\cite{icalp18}. In the following, 
we hence use Algorithm~\ref{alg:main}
as a baseline and lower bound on the number of rounds needed
by the more complex algorithm in~\cite{icalp18}.

The input data does not provide any capacity on the links.
Capacities determine the block dependency and an insufficient capacity
allocation can lead to a cyclic $D$
which renders an instance infeasible. On the other hand,
examining all possible allocations is not practical.
Hence, in order to capture the maximum rounds in each graph efficiently,
we take into account also the infeasible instances.
Later, we explain how infeasible instances are handled.
In a preprocessing step, the evaluation takes the raw graph and
allocates minimal capacities:
set the capacity to 2 for links that carry the old/update flow paths of both pairs,
otherwise set the capacity to 1.

The program, for every pair of source and destination $(s,t)$, first computes all the paths from $s$ to $t$.
Next, it iterates over all possible path pairs (i.e.~old and update paths) chosen independently for each of the two flows (dismissing identical path pairs).
Each iteration does the following.
\begin{enumerate}
	\item 
	Perform Line~\ref{line:1} on the path pairs and generates a block dependency graph $D$.
	\item
	Enumerate all paths in $D$ and each path $P$ is weighted as follows.
	\begin{enumerate}
		\item Initialize $w(P) = |P|$.
		
		\item \label{item:preparation_round}
		Let $b_1$ be the block that corresponds to the last vertex of $P$. 
		Set $w(P) = w(P)+1$  if $|E[b_1 \cap F^u(b_1)]| > 1$. 
		
		\item \label{item:cleanup_round}
		Let $b_2$ be the block that corresponds to the first vertex of $P$. 
		Set $w(P) = w(P)+1$  if $|E[b_1 \cap F^o(b_1)]| > 1$. 
	\end{enumerate}

	\item 
	Find the path $P_{max} = \max_{P'} w(P')$  (ties broken arbitrarily). Let $\rounds = w(P_{max})$.

	\item \label{itm:singleblocks}
	For each block $b$ corresponding to a vertex in $D$ apply
	$\rounds = \max(\rounds, |E[b \cap F^o(b)]| + |E[b \cap F^u(b)]| + 1)$.	 
\end{enumerate}
At the end, $\rounds$ will hold the actual number of rounds it takes in the optimal
schedule produced by \Cref{alg:mainopt}.
The case \ref{item:preparation_round} accounts for the  preparation (i.e.~adding new flow rules)
round of the  block scheduled earliest in a chain of dependent blocks (i.e.~ current path $P$).
Similarly, \ref{item:cleanup_round} accounts for the  cleanup round (i.e.~removal of old flow rules) of the block
scheduled the latest in that chain.

Eventually, $\rounds$ is determined  either by the chain of dependent blocks that
corresponds to the longest weighted path in $D$,
or by some single block at Line \ref{itm:singleblocks}
due to extra preparation/cleanup rounds consumed by that block.

Any infeasible instance, i.e.~with cyclic block dependency,
can be turned feasible by increasing the capacity of
some link from 1 to 2, hence breaking the cycle.
Therefore, starting from minimal capacity allocation
is always sufficient to preserve the worst case in any topology.

The results show that the optimal number of rounds on the
subject networks vary between 2 and 6 (see \Cref{fig:roundschart_opt}).
Table~\ref{tbl:zoo} lists the numbers obtained from some of
the examined graphs (numbers are rounded to the nearest integer, except those close to 0).

In the second implementation (i.e.~feasibility only),
 we evaluated \Cref{alg:main} on the same
 input data and obtained feasible schedules under
worst case capacity allocations.
The number of rounds spread between 2 and 14 (see \Cref{fig:roundschart_feas}).

We also implemented the MIP in \Cref{MIP}.
The runtime even for 2 flows are usually in few seconds,
 much longer compared to the results from the first implementation
 ($\approx100$ microseconds, see \Cref{fig:times_plot} and \Cref{fig:example_mip}).

\begin{figure}[!h]
	\begin{subfigure}{.33\columnwidth}
		\centering
		\includegraphics[width=0.99\columnwidth]{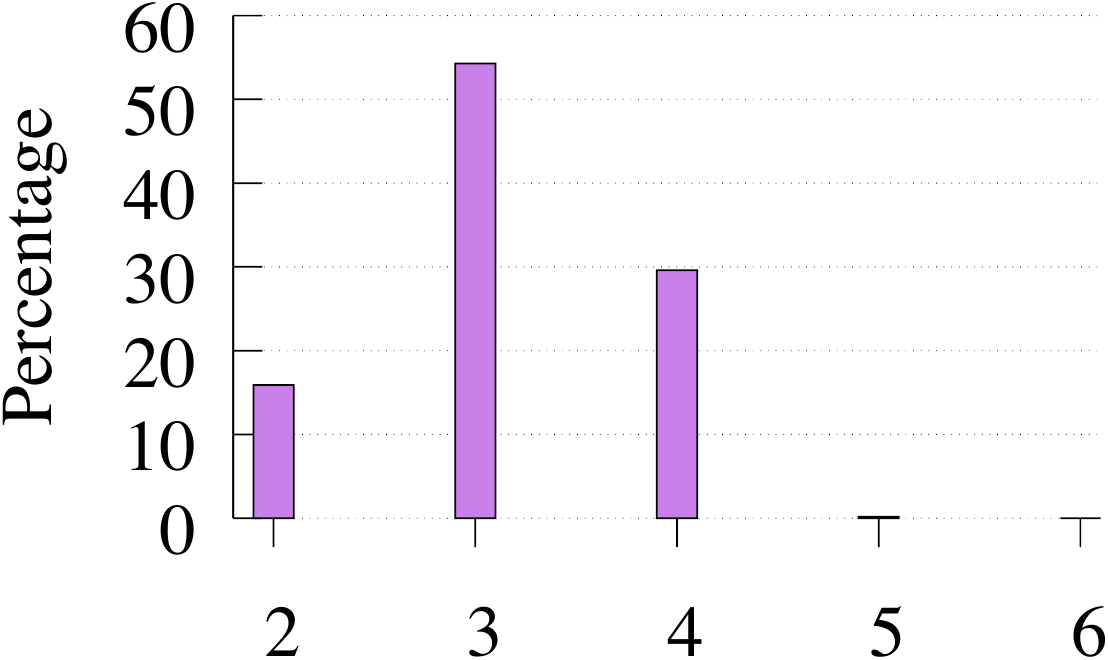}
		\caption{}
		\label{fig:roundschart_opt}
	\end{subfigure}%
	\begin{subfigure}{.33\columnwidth}
		\centering
		\includegraphics[width=0.99\columnwidth]{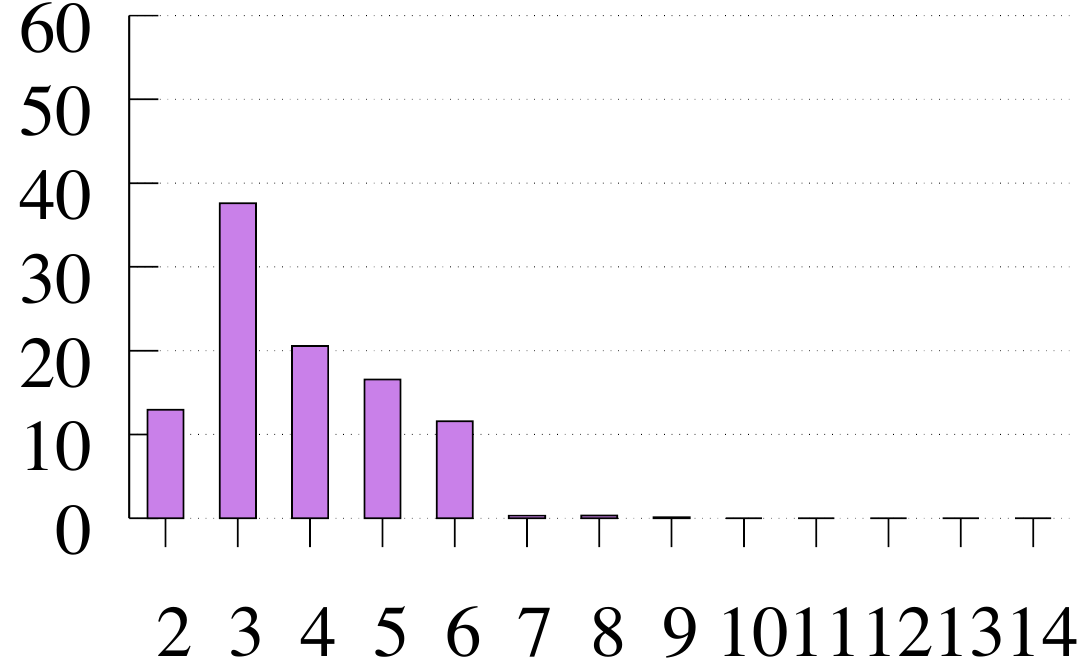}
		\caption{}
		\label{fig:roundschart_feas}
	\end{subfigure}%
	\begin{subfigure}{.33\columnwidth}
		\tabcolsep=0.06cm
			\begin{tabular}{l c c c c c}
				\hline\hline 
				Graph & $2$ & $3$ & $4$ & $5$ & $6$  \\ [0.5ex]
				\hline 
				AARNET & 46 & 50 & 4 & 0 & 0	\\ 
				Abilene & 37 & 57 & 6 & 0 & 0	\\ 
				Belnet(2003) & 13 & 79 & 8 & 0.01 & 0	\\ 
				GARR(2011) & 21 & 68 & 10 & 1 & 6e-4	\\ 
				SWITCH & 10 & 77 & 13 & 0.03 & 0	\\ 
				\hline 
			\end{tabular}
		\caption{}
		\label{tbl:zoo}
	\end{subfigure}
	\caption{
		(\subref{fig:roundschart_opt}) The frequency of each possible number of rounds (in \%) in optimal schedules from \Cref{alg:mainopt},
		 (\subref{fig:roundschart_feas}) in arbitrary feasible schedules from \Cref{alg:main},
		and (\subref{tbl:zoo}) in optimal schedules for a sample of the examined graphs.
	}
	\label{fig:charts_rounds}
\end{figure}

\begin{figure}[!h]
	\begin{subfigure}{.5\columnwidth}
		\includegraphics[width=0.99\columnwidth]{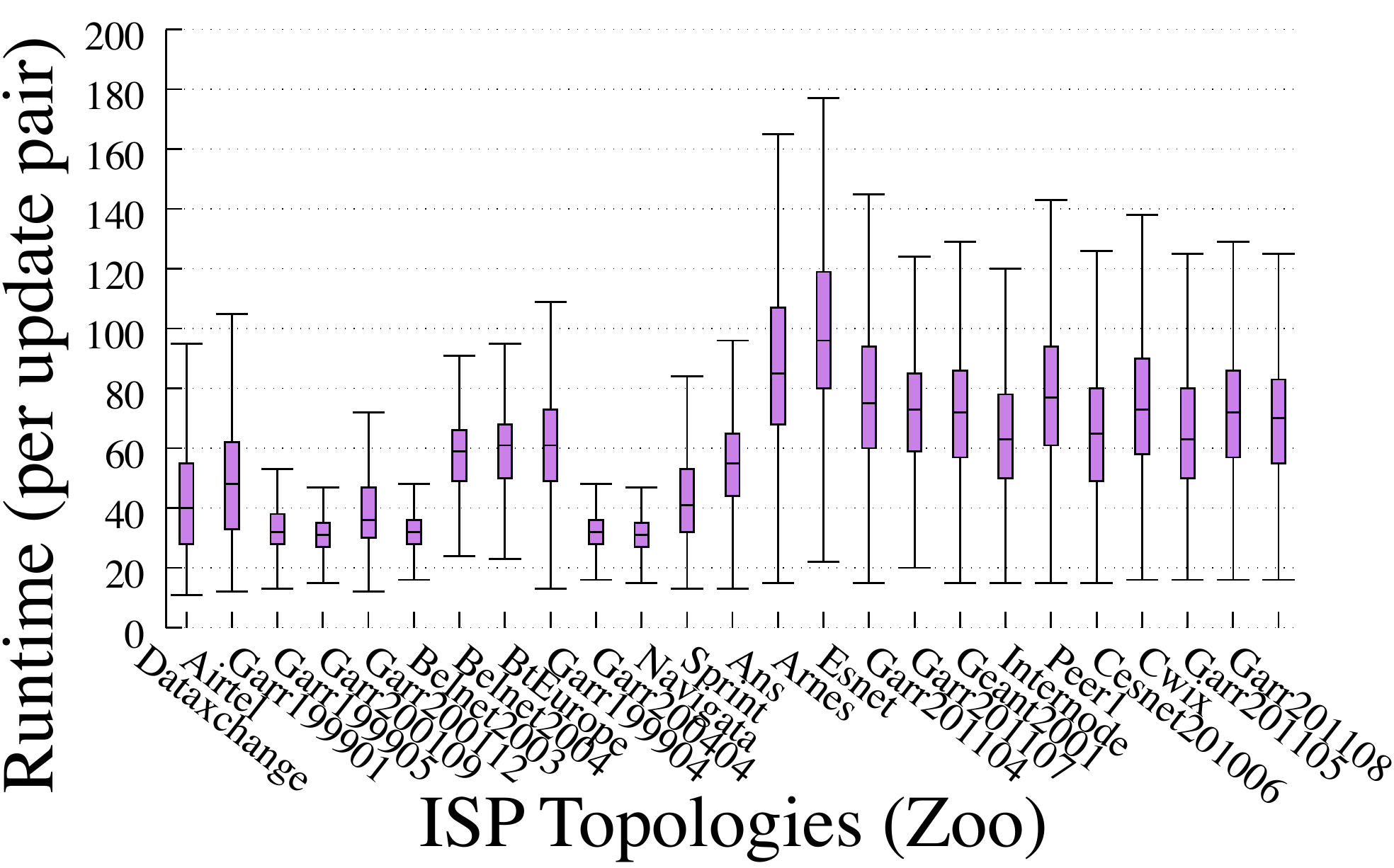}
		\caption{}
		\label{fig:times_plot}
	\end{subfigure}%
	\begin{subfigure}{.5\columnwidth}
		\includegraphics[width=0.99\columnwidth]{_DialtelecomCz.pdf}
		\caption{}		
		\label{fig:example_mip}
	\end{subfigure}	
		\caption{
			(\subref{fig:times_plot}) Distribution of runtime (in microsecond) over all the
			problem instances from some of the evaluated graphs.
			(\subref{fig:example_mip}) The graph (Dial Telecom) with a long runtime in our MIP implementation
			(the worst case takes $\approx26$ seconds)
		}
\end{figure}

\section{Related Work}\label{sec:relwork}

The fundamental problem of how to reroute flows
has recently received much attention in the networking
community and we refer the reader to the recent survey by 
Foerster et al.~\cite{update-survey} for an overview
of the field. 
Yet, today, and in contrast to the classic problem of how
to \emph{route} flows~\cite{AmiriKMR16,shimon-flows,KawarabayashiKK14,ufkleinberg,leighton1999multicommodity,ufskutella},
we still know surprisingly little about useful 
\emph{algorithmic} techniques for efficient 
flow \emph{rerouting}. 

There exist several empirical studies
motivating our model~\cite{dionysus,kuzniar2015you},
however, this literature is orthogonal to ours.
Moreover, many existing consistent network update algorithms
such as~\cite{infocom15,dionysus,zupdate,roger,abstractions}
require packet tagging and additional forwarding rules, which
render the problem different in nature.
Mahajan and Wattenhofer~\cite{roger}
initiated the study of flow rerouting algorithms which schedule
updates over time. The authors also presented first algorithms
to quickly updates routes in a transiently \emph{loop-free} manner~\cite{sirocco16update,Forster2016Consistent,Forster2016Power},
by maximizing \emph{the number of updates per round}.
A second line of research focuses on \emph{minimizing the number of rounds}
of loop-free updates~\cite{dsn16,sigmetrics16,ludwig2015scheduling,hotnets14update}.
	
As congestion is known to 
negatively affect application performance and
user experience, it has also been studied 
intensively in the context of flow rerouting problems.
	The seminal work by Hongqiang et al.~\cite{zupdate} 
	on congestion-free rerouting
	has already been extended in several papers, 
	using static~\cite{roger-infocom,swan,jaq4,icnp-jiaqi}, 
	dynamic~\cite{jaq1}, or time-based~\cite{jaq8,jaq10} 
	approaches. 
	Vissicchio et al.~presented FLIP~\cite{vissicchio2016flip},
	which combines
	per-packet consistent updates with order-based rule replacements, 
	in order to reduce memory overhead:
	additional rules are used only when necessary.
	Moreover, Hua et al.~\cite{huafoum}
	recently initiated the study of 
	adversarial settings,
	and presented FOUM, 
	a flow-ordered
	update mechanism that is robust to packet-tampering and packet dropping
	attacks. 
However, none of these papers present polynomial-time
algorithms for rerouting flows without requiring packet tagging.

Our work on polynomial-time algorithms is motivated in particular by 
the negative result by Ludwig et al.~\cite{ludwig2015scheduling}
who showed that deciding whether a loop-free 3-round
update schedule exists is NP-hard, even in the absence of
capacity constraints.
Given this negative result, much prior work typically resorts
to heuristics~\cite{icdcs17}, which however do not come with any
formal guarantees on the quality of the computed
schedule, or to algorithms which have a super-polynomial
runtime~\cite{sigmetrics16}. 
The only exception is the polynomial-time algorithm
by Amiri et al.~\cite{icalp18} for acyclic flow graphs, 
which however is limited to 
computing \emph{feasible} (possibly very long)
update schedules.
There are various differences between this paper and the recent work of
Amiri et al.:
(1) In contrast to our work where only flow pairs
need to form a DAG, \cite{icalp18} considers a much more restricted
model where the union of all flows must be acyclic. 
This restriction allows the authors to design an FPT
  algorithm for $k$ flows, whereas in our model 
  the problem is  NP-complete already for six flows. 
  Hence, different techniques are required to show hardness.
(2) Similarly to~\cite{icalp18}, our algorithm relies on 
a dependency graph that
  explains the relation between flows. However,
  since we aim to compute schedules only for two flows, 
we do not require the big machinery introduced for
  $k$ flows instead we provide a more elegant algorithm.
(3) At the same time, since in contrast to prior work, we focus on an optimal solution, 
our model is more chalelnging and requires new algorithmic ideas.

 Finally, our problem is situated in the larger context of
 combinatorial reconfiguration theory, which has recently
 received much attention, e.g., in the context of games~\cite{van2013complexity}.
  In this respect, 
 the reconfiguration model closest to ours is by 
 Bonsma~\cite{bonsma2013complexity} 
 	who studied
 	how to perform rerouting such that transient paths are always \emph{shortest}.
 	However, the corresponding techniques and results are not applicable
 	in our model where we consider flows of certain \emph{demands},
 	and where different flows may \emph{interfere} due to capacity constraints
 	in the underlying network.


\section{Conclusion and Future Work}\label{sec:conclusion}

This paper presented the first polynomial-time and optimal scheduling algorithm to 
quickly reroute two flows in a congestion-free manner. 
In particular, the algorithm can be used to minimize the number of required (asynchronous) interactions between switches and controller in a software-defined network. 
We also prove that our result cannot be generalized much further 
as the problem becomes NP-hard
already for six flows. 
One of the main open question of our work concerns the polynomial-time
tractability for $2<k<6$ flows. These cases 
 might be very challenging, and currently, we do not have any insights
on how to deal even with three flows.

In this paper we assumed that every pair of flows
forms a DAG, and we did not constrain the underlying network topology
which can be general. 
However, many real-world networks are sparse, and feature 
nice topological or combinatorial structures. The study of such
networks introduces another interesting direction for future research.
We currently do not know what is the complexity of the problem even on planar
graphs, one of the most interesting classes of sparse graphs. 

\bigskip
\noindent \textbf{Acknowledgements.} We would like
to thank Stephan Kreutzer, Arne Ludwig, 
 and Roman Rabinovich for discussions on this
problem.
The research of Saeed Akhoondian Amiri and Sebastian Wiederrecht has been supported by 
the European Research Council (ERC) under the European Union's Horizon
2020 research and innovation programme (ERC consolidator grant DISTRUCT,
agreement No 648527). 

\balance

{
	\bibliography{literature}
}


\end{document}